\documentclass[11pt, letterpaper]{article}
\usepackage[letterpaper, margin=1in]{geometry}
\usepackage{amsmath,amssymb,amsthm}
\usepackage{todonotes,comment}
\usepackage{titling}
\usepackage{float}
\usepackage{bbm}
\usepackage[ruled,linesnumbered, vlined]{algorithm2e}
\usepackage{thm-restate}
\usepackage{hyperref}

\newtheorem{thm}{Theorem}
\numberwithin{thm}{section}

\newtheorem{bigthm}{Theorem}
\newtheorem{lemma}[thm]{Lemma}
\newtheorem{proposition}[thm]{Proposition}
\newtheorem{remark}[thm]{Remark}
\newtheorem{defn}[thm]{Definition}
\newtheorem{claim}[thm]{Claim}

\def\polylog{\operatorname{polylog}}
\def\poly{\operatorname{poly}}

\title{Weighted Matching in the Random-Order Streaming and \\ Robust Communication Models}
\author{Diba Hashemi \\EPFL \and Weronika Wrzos-Kaminska \\EPFL }
\date{}
\begin{document}
\begin{titlingpage}
\maketitle
\begin{abstract}
    We study the maximum weight matching problem in the random-order semi-streaming model and in the robust communication model.  Unlike many other sublinear models, in these two frameworks, there is a large gap between the guarantees of the best known algorithms for the unweighted and weighted versions of the problem.  
    
    In the random-order semi-streaming setting, the edges of an $n$-vertex graph arrive in a stream in a random order. The goal is to compute an approximate maximum weight matching with a single pass over the stream using $O(n\polylog n)$ space. Our main result is a $(2/3-\epsilon)$-approximation algorithm for maximum weight matching in random-order streams, using space $O(n \log n \log R)$, where $R$ is the ratio between the heaviest and the lightest edge in the graph. Our result nearly matches the best known unweighted $(2/3+\epsilon_0)$-approximation (where $\epsilon_0 \sim 10^{-14}$ is a small constant) achieved by Assadi and Behnezhad [ICALP 2021], and significantly improves upon previous weighted results. 
    
    Our techniques also extend to the related robust communication model, in which the edges of a graph are partitioned randomly between Alice and Bob. Alice sends a single message of size $O(n\polylog n)$ to Bob, who must compute an approximate maximum weight matching. We achieve a $(5/6-\epsilon)$-approximation using $O(n \log n \log R)$ words of communication, matching the results of Azarmehr and Behnezhad [ICALP 2023] for unweighted graphs.
\end{abstract}
\end{titlingpage}
\section{Introduction}

The maximum matching problem is a fundamental problem in graph algorithms. In the unweighted version of the problem, we are interested in computing a maximum \emph{cardinality} matching, i.e. to maximize the total number of edges in the matching. In the weighted version, we are interested in computing a maximum \emph{weight} matching, i.e. to maximize the sum of the edge weights in the matching. 

In this paper, we study matchings in the semi-streaming model. The semi-streaming model, originally introduced in \cite{Stream_OG}, is motivated by the rise of massive graphs where the data is too large to be stored in memory, and has received extensive attention (see among others \cite{McG05, ELMS11, GKK, Zelke12, Kap_lb_mid, CS14, 1/2-adversarial, simplified_1/2, Kap_lb_best}). In this model, the edges of a graph arrive sequentially as a stream. The algorithm typically makes a single pass over the stream using space $O(n \polylog n)$ and must output an approximate maximum matching at the end of the stream. If the graph is unweighted, the greedy algorithm trivially gives a $1/2$-approximation, which is the best known for adversarially ordered streams. On the hardness side, it is known that a $0.59$-approximation is not possible \cite{Kap_lb_best} (see also \cite{GKK, Kap_lb_mid}). Closing the gap between these upper and lower bounds is one of the major open problems in the graph streaming literature. There has also been a long line of work on the weighted problem \cite{Stream_OG, McG05, ELMS11,Zelke12, CS14, 1/2-adversarial, simplified_1/2}, culminating in a $(1/2 - \epsilon)$-approximation using space $O(n)$ \cite{1/2-adversarial, simplified_1/2}. 

Recently, there has been a wide interest in the random-order version of this problem, in which the arrival order of the edges is chosen uniformly at random. This problem has been extensively studied in the unweighted setting \cite{KMM12, konrad_stream,  stream_ABBMS, Ola,Bernstein,  stream_Farhadi,  beating_two-thirds, AS23}. Notably, Bernstein \cite{Bernstein} gave a $2/3$-approximation, and Assadi and Behnezhad \cite{beating_two-thirds} improved it to $(2/3 + \epsilon_0)$ for a small constant $\epsilon_0 \sim 10^{-14}$.

Progress on the weighted version of the problem lags behind. Gamlath et al. \cite{Ola} broke the barrier of $1/2$ in weighted graphs by obtaining a $(1/2+\delta)$-approximation for a small constant $\delta \sim 10^{-17}$. More recently, Huang and Sellier \cite{b-matching} gave a $\frac{1}{2-1/(2W)}$-approximation under the assumption that the weights take integral values in $[W]$. This leaves a considerable gap between the best known results for the unweighted and weighted versions of the problem. In contrast, in other sublinear contexts, such as adversarially ordered streams or the dynamic graph setting, the weighted/unweighted gap has largely been closed \cite{BDL}. The challenge of closing the gap in random-order streams remains an open problem, and has been highlighted explicitly in \cite{Bernstein} and \cite{BDL}. 

In this paper, we give a $(\frac{2}{3}-\epsilon)$-approximation algorithm for the weighted setting. Our result almost matches the best known $(\frac{2}{3} + \epsilon_0)$-guarantee for the unweighted setting and improves significantly upon the previous results for the weighted setting. 

 \begin{bigthm}\label{thm:main_stream}
Given any constant $\epsilon>0$, there exists a deterministic single-pass streaming algorithm that with high probability computes a  $(\frac{2}{3}-\epsilon)$-approximate maximum weight matching if the edges arrive in a uniformly random order. The space usage of the algorithm is $O(n\log n \log R),$ where $R$ is the ratio between the heaviest and the lightest edge weight in the graph. 
 \end{bigthm}
 
 We also consider the two-player communication complexity model \cite{CommComplexity_OG}, and in particular the one-way communication complexity of matching, first studied in \cite{GKK}. Here, the edge-set is partitioned between two parties Alice and Bob. Alice sends a single message to Bob, who must output an approximate maximum matching. Typically, we are interested in protocols with communication complexity $O(n \polylog n).$ 
 
 If the edges are partitioned adversarially between Alice and Bob, the right answer turns out to be $2/3$. A $2/3$-approximation can be achieved using $O(n)$ communication for both bipartite unweighted \cite{GKK}, general unweighted \cite{AB} and general weighted \cite{BDL} graphs. Going beyond a $2/3$-approximation requires $n^{1+1/(\log \log n)} \gg n \polylog n$ communication even for unweighted bipartite graphs \cite{GKK}. 
 
If instead the edges are partitioned randomly between the two parties, the answer is less clear. Recently, Azarmehr and Behnezhad \cite{RobustComm} gave a $5/6$-approximation algorithm for unweighted graphs, improving upon a previous result of Assadi and Behnezhad \cite{RobustComm_bipartite}. To the best of our knowledge, prior to our work, there were no results for weighted graphs (besides the $2/3$-approximation implied by adversarial protocols). We match the unweighted guarantees of Azarmehr and Behnezhad \cite{RobustComm}, thus closing the weighted/unweighted gap in the robust communication model. 
 
 \begin{bigthm}\label{thm:main_2party}
Given any constant $\epsilon>0$, there exists a protocol that with high probability computes a $(\frac{5}{6}-\epsilon)$-approximate maximum weight matching in the two-party robust communication model using  $O(n\log n \log R)$ words of communication, where $R$ is the ratio between the heaviest and the lightest edge weight in the graph. 
 \end{bigthm}
More generally, we match the results of Azarmehr and Behnezhad \cite{RobustComm} for unweighted $k$-party robust communication, thus closing the unweighted/weighted gap also in this model. 
 \begin{bigthm}\label{thm:main_kparty}
Given any $k \geq 2$ and any constant $\epsilon>0$, there exists a protocol that with high probability computes a $(\frac{2}{3}+\frac{1}{3k}-\epsilon)$-approximate maximum weight matching in the $k$-party one-way robust communication model using  $O(n\log n \log R)$ words of communication, where $R$ is the ratio between the heaviest and the lightest edge weight in the graph. 
 \end{bigthm}

\subsection{Related Work}The maximum matching problem is one of the most studied problems in the streaming setting, with numerous lines of work. This includes among others single-pass algorithms \cite{Stream_OG, McG05, ELMS11, GKK, Zelke12,Kap_lb_mid,  CS14,1/2-adversarial, simplified_1/2,  Kap_lb_best, Regularitylemma}, multi-pass algorithms using $2$ or $3$ passes \cite{KMM12,EHM16, KT17, konrad_stream, KN21,  Assadi22, FS22, KNS23, KN24}, and $(1-\epsilon)$-approximation using a higher number of passes \cite{McG05, EKMS12,AG13, GO16,AG18, Tir18, Ola, ALT21, AJJST22, FMU22, AS23, HS23, Assadi24}. Garg et al. considered matching in a robust random-order streaming model with adversarial noise \cite{Ola_robust}. There are many results on dynamic streams, where edges can be deleted \cite{CCHM15,Kon15,  AKLY16, CCEHMMV16, AKL17, DK20, AS22}. A different line of work considers estimating matching size, either in random-order streams \cite{KKS14, BS15, MMPS17, KMNT20, AS23_exactmatching} or in adversarially ordered streams \cite{BS15, MV16,AKL17,  CJMM17, EHLMO18,MV18, BGMMSVZ19, AKSY20, AN21}.  Finally, there have also been several works on exact matching \cite{Stream_OG, CCEHMMV16, GO16, AR20, KPS0Y21, AJJST22}. 

\section{Technical Overview}\label{sec:tech}
In this paper, we are interested in the random-order streaming model. The maximum cardinality matching problem has gained significant attention within this framework  \cite{KMM12, konrad_stream, Ola, stream_ABBMS, stream_Farhadi, Bernstein,  beating_two-thirds}.  Bernstein \cite{Bernstein} gave a $2/3$-approximation algorithm by adapting the  ``matching sparsifier" Edge-Degree Constrained Subgraph  (EDCS) to the streaming context. Subsequent work by Assadi and Behnezhad \cite{beating_two-thirds} improved upon this, achieving a $(2/3 + \epsilon_0)$-approximation by simultaneously running Bernstein's algorithm while identifying short augmenting paths. One of the motivations for studying the random-order setting, is that real-world data is rarely ordered adversarially. Rather, in most practical applications, it is reasonable to assume that the data is drawn from some distribution. However, assuming uniform randomness is often too strong of an assumption since data correlations are prevalent in many real-world settings. This raises the question: 

\begin{center}
    \em 
How robust are random-order streaming algorithms to correlations in the arrival order? 
\end{center}
The robustness of random-order streaming algorithms to various types of adversarial distortions has already been studied previously, among others in the context of maximum matching and submodular maximization \cite{Ola_robust}, rank selection \cite{semirandom_quantile1, semirandom_quantile2, semirandom_quantile3}, clustering problems \cite{semirandom_facility} and component collection and counting \cite{CKKP}. 
 In this paper, we focus on matchings. Our first contribution is showing that existing algorithms for unweighted matching in random-order streams are in fact robust to correlations in the arrival order. 
 \begin{center}
    \em 
Berstein's $\left( \frac{2}{3}-\epsilon \right)$-approximation algorithm is resilient to (limited) adversarial correlations in the arrival order. 
\end{center}
 
Surprisingly, this immediately gives a reduction from weighted matching in random-order streams.

In \emph{adversarially} ordered streams, Bernstein, Dudeja and Langley \cite{BDL}  gave a reduction from maximum weight matching to maximum cardinality matching. Progress in \emph{random-order streams} has been comparatively limited. Gamlath et al. \cite{Ola} achieved a $(1/2+\delta)$-approximation, where  $\delta \sim 10^{-17}$ is a small constant. More recently, Huang and Sellier \cite{b-matching} gave a $\frac{1}{2-1/(2W)}$-approximation under the assumption that the weights take integral values in $[W]$, improving upon the result of Gamlath et al. \cite{Ola} for small weights. They generalized the definition of EDCS to weighted graphs, which enabled them to adapt Bernstein's algorithm \cite{Bernstein} to weighted graphs. However, their generalized notion of EDCS has weaker guarantees compared to the unweighted version, resulting in a significant loss in the approximation ratio. 
 
Our second contribution is to nearly close the gap between weighted and unweighted maximum matching in random-order streams. We show that the reduction of Berstein, Dudeja and Langley can be applied to random-order streaming algorithms which are resilient to specific correlations in the arrival order. This, together with the fact that Bernstein's algorithm \cite{Bernstein} is robust to the appropriate correlations, gives a $2/3$-approximation algorithm for weighted bipartite graphs. We are also able to extend the guarantees to non-bipartite graphs.

\subsection{Reduction in Adversarial Streams} First, we review the reduction of Bernstein, Dudeja and Langley \cite{BDL} for adversarial streams. It is based on a technique called graph unfolding by Kao, Lam, Sung and Ting \cite{KLST}.
\begin{defn}[Graph Unfolding \cite{KLST}]
	Let $G = (V, E, w)$ be a graph with non-negative integral edge weights. The unfolded graph $\phi(G)$ is an unweighted graph created as follows. For each vertex $u \in V$, let $W_u =max_{e \ni u} w_e$ be the maximum edge weight incident on $u$. There are $W_u$ copies of $u$ in $\phi(G)$, denoted by $u^1, ..., u^{W_u}$. For each edge $e= (u,v)$ in $G$, there are $w_e$ edges $\{(u^{i},v^{w_e-i+1}) \}_{i \in [w_e]}$ in $\phi(G)$. See Figure~\ref{fig:unfolding} for an illustration.
\end{defn}
\begin{figure}[H]
	\centering
 \resizebox{0.6\linewidth}{!}{%
	\begin{tikzpicture}
\node (G) at (1,2) {$G$};
\begin{scope}[every node/.style={draw, circle, fill=white, inner sep=1.5pt}]
    \node (x) at (1,-0.5) {$x$};
    \node (y) at (0,1) {$y$};
    \node (z) at (2,1) {$z$};
\end{scope}
    \draw (x) -- (y) node[midway, above]{$2$};
    \draw (x) -- (z) node[midway, above]{$3$}; 
\draw[->] (2.5,0.2) --(4.3,0.2) node[midway, above]{$\phi$};
\node (phi) at (6.3,2) {$\phi(G)$};
\begin{scope}[every node/.style={draw, circle, fill=white, inner sep=0pt,minimum size=1pt}]
\node (x1) at (5.3,-0.5) {$x^1$};
\node (x2) at (6,-0.5) {$x^2$};
\node (x3) at (6.7,-0.5) {$x^3$};
\node (y1) at (4.8,1) {$y^1$};
\node (y2) at (5.5,1) {$y^2$};
\node (z1) at (6.3,1) {$z^1$};
\node (z2) at (7,1) {$z^2$};
\node (z3) at (7.7,1) {$z^3$};
\end{scope}
    \draw (x1) -- (y2);
    \draw (x2) -- (y1);
    \draw (x1) -- (z3); 
    \draw (x2) -- (z2); 
    \draw (x3) -- (z1); 
\end{tikzpicture}
 }
	\caption{An example of a weighted graph $G$ and its unfolding $\phi(G)$.}
	\label{fig:unfolding}
\end{figure}
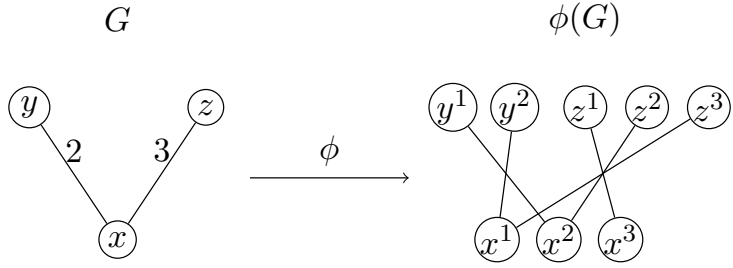
One can also do a reverse operation of unfolding to bring a subgraph back to $G$.
\begin{defn}[Refolding \cite{BDL}]
	Let $G=(V,E)$ be a weighted graph and let $H \subseteq \phi(G)$. The refolded graph $\mathcal{R}(H)$ has
	vertex set $V$ and edge set $E(\mathcal{R}(H)) := \{e= (u,v) \in G : (u^{i},v^{w_e-i+1}) \in H$ for some $i \in [w_e]\}$. See Figure~\ref{fig:refolding} for an illustration.
\end{defn}
\begin{figure}[H]
	\centering
     \resizebox{0.6\linewidth}{!}{%
	\begin{tikzpicture}
\node (G) at (1,2) {$G$};
\begin{scope}[every node/.style={draw, circle, fill=white, inner sep=1.5pt}]
    \node (x) at (1,-0.5) {$x$};
    \node (y) at (0,1) {$y$};
    \node (z) at (2,1) {$z$};
\end{scope}
    \draw[blue, line width = 3pt] (x) -- (y) node[midway, above,black]{$2$};
    \draw (x) -- (z) node[midway, above]{$3$}; 
    \node[blue] (R text) at (-0.1,0.1){$\mathcal{R}(H)$};
\draw[->] (4.3,0.2)-- (2.5,0.2) node[midway, above]{$\mathcal{R}$};
\node (phi) at (6.3,2) {$\phi(G)$};
\begin{scope}[every node/.style={draw, circle, fill=white, inner sep=0pt,minimum size=1pt}]
\node (x1) at (5.3,-0.5) {$x^1$};
\node (x2) at (6,-0.5) {$x^2$};
\node (x3) at (6.7,-0.5) {$x^3$};
\node (y1) at (4.8,1) {$y^1$};
\node (y2) at (5.5,1) {$y^2$};
\node (z1) at (6.3,1) {$z^1$};
\node (z2) at (7,1) {$z^2$};
\node (z3) at (7.7,1) {$z^3$};
\end{scope}
    \draw[blue, line width = 3pt] (x1) -- (y2);
    \draw (x2) -- (y1);
    \draw (x1) -- (z3); 
    \draw (x2) -- (z2); 
    \draw (x3) -- (z1); 
    \node[blue] (text H) at (5,0.1){$H$};
\end{tikzpicture}
 }
	\caption{An example of a subgraph $H \subseteq \phi(G)$ and its refolding $\mathcal{R}(H) \subseteq G$. In this example, $H = \{(u^1,v^2)\}$. Then $\mathcal{R}(H) = \{(u,v)\}$.}
	\label{fig:refolding}
\end{figure}
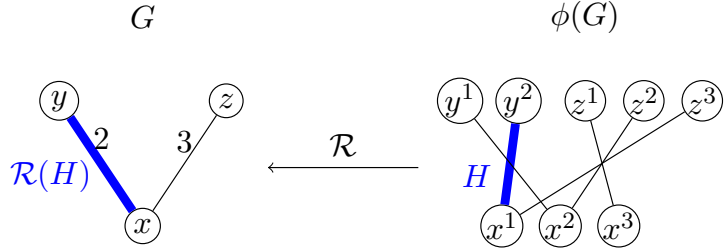
Figure \ref{fig:unfolding} illustrates the unfolding operation and Figure \ref{fig:refolding} illustrates the refolding operation. The key property of refolding is that it preserves the matching size in bipartite graphs. 
\begin{lemma}[\cite{BDL}]
	\label{lemma:matchsize}
	Let $G$ be a weighted bipartite graph, and let $H \subseteq \phi(G)$ be a subgraph of its unfolding. Then $\mu_w(\mathcal{R}(H)) \geq \mu(H)$.
\end{lemma} 
This leads to a reduction from maximum weight bipartite matching to maximum cardinality bipartite matching in \emph{adversarially ordered streams}: Upon the arrival of each weighted edge $e \in G$,  unfold $e$ and pass the corresponding unweighted edges $\phi(e)$ into an unweighted streaming algorithm. At the end of the stream, we obtain an unweighted matching in $\phi(G)$, which we can refold to obtain a weighted matching in $G$. 

In \emph{random-order streams}, this reduction breaks for the following reason: For each weighted edge $e \in G$, the unweighted edges $\phi(e)$ will necessarily arrive together. This introduces correlations in the arrival order of the edges, so the guarantees of random-order streaming algorithms do not apply. To address this, we consider a new streaming model, the $b$-batch random-order stream model, which is similar to the hidden-batch model introduced in \cite{CKKP}. This model allows us to capture the edge-correlations that arise from graph unfolding. 
\pagebreak

\begin{restatable}[$b$-batch random-order stream model]{defn}{bbatch}
	In the $b$-batch random-order stream model the edge set of the input graph $G = (V, E)$ is presented as follows: 
	An adversary partitions the edge set $E$ into batches $\mathcal{B} = \{B_1, ..., B_q\}$ with $|B_i| \leq b$ for all $i$. The arrival order of the batches $(B_{i_1}, ..., B_{i_q})$ is then chosen uniformly at random among all the permutations of $\mathcal{B}$. The edges in each batch arrive simultaneously. 
\end{restatable}

 The graph unfolding technique gives a reduction from weighted bipartite random-order streams to unweighted bipartite $b$-batch random-order streams. Indeed, each batch corresponds to one weighted edge. So given a weighted graph $G$, we can simply run a $b$-batch random-order stream algorithm on $\phi(G)$ with batches $\mathcal{B} = \{\phi(e): e \in G\}. $

\subsection{Bernstein's Algorithm for Unweighted Random-Order Streams} We now review Bernstein's algorithm for unweighted random-order streams \cite{Bernstein}. The algorithm proceeds in two Phases. Let $\beta = O(\poly(\epsilon^{-1}))$ be a parameter. Phase 1 constructs a subgraph $H$ such that for all $(u,v) \in H$,
\begin{equation}\label{eqn:underfull}
\deg_H(u) + \deg_H(v) \leq \beta.
\end{equation}
    Given a subgraph $H$, we will say that an edge $(u,v) \in G$ is \emph{underfull} if  $\deg_H(u) + \deg_H(v) \leq \beta -2$, 
otherwise say that $(u,v)$ is non-underfull.

 The algorithm constructs $H$ by greedily adding underfull edges, and then removing any edges that violate Equation \ref{eqn:underfull}. Phase 1 terminates when $\approx \poly(\epsilon) \frac{m}{n}$ non-underfull edges arrive in a row, and the algorithm then moves on to Phase 2. Bernstein \cite{Bernstein} showed that it is only possible to make at most $n \beta^2$ modifications to $H$. Since Phase 1 terminates when we see $\approx \poly(\epsilon) \frac{m}{n}$ edges in a row without modifying $H$, the phase must terminate within the first $\approx n \beta^2 \cdot  \poly(\epsilon) \frac{m}{n} \approx \epsilon m$ edges. This argument also holds in the $b$-batch random-order stream model. 
 
Then, in Phase 2, the algorithm simply collects all underfull edges into a separate set $U$ (without modifying the graph $H$). Let $G_{late}$ denote the edges that arrive in Phase 2. Bernstein \cite{Bernstein} proved the following structural result about $H \cup U$, which holds regardless of the assumptions on the arrival order: 
\begin{equation}\label{eqn:Bguarantee}
	\mu(H \cup U) \geq \left(\frac{2}{3}-\epsilon\right) \mu(G_{late}). 
\end{equation}
Since Phase 2 contains at least a $(1-\epsilon)$ fraction of the edges, and since the stream is uniformly at random, it follows from Chernoff bounds that $\mu(G_{late}) \geq (1-2\epsilon)\mu(G)$. Consequently, by Equation \ref{eqn:Bguarantee}, it holds that
$$\mu(H \cup U) \geq \left(\frac{2}{3}-3\epsilon\right)\mu(G).$$

For the space analysis, observe that $H$ contains at most $n \beta = O(n)$ edges. Let us now consider $U$. Recall that $U$ is the set of all underfull edges that arrive after the termination of Phase 1, and that Phase 1 terminates when we see $\approx \frac{m}{n}$ non-underfull edges in a row. So the only way for $U$ to become too large is if we draw $\approx \frac{m}{n}$ non-underfull edges in a row when there are more than $C \cdot n \log n$ underfull edges left in the stream, for some constant $C$. The probability of this event can be upper-bounded by 
\begin{equation*}
	\left(1- \frac{C \cdot n \log n}{m}\right)^{m/n} \leq n^{-C},
\end{equation*}
so with high probability, the algorithm stores at most $O(n \log n)$ edges. Note that the space analysis breaks down in the $b$-batch random-order stream model, due to the correlated arrival orders. 

\subsection{Applying the Algorithm to Batch Arrivals} We now sketch why Bernstein's algorithm can be adapted to work under batch arrivals. Let $b$ denote the upper bound on the batch size, and let $q$ denote the total number of batches in the stream. Recall that in the reduction from weighted random-order streams, $b$ corresponds to the maximum weight in the graph and $q$ corresponds to the number of edges in the weighted graph. We will now describe how to obtain an algorithm with a polynomial space dependence on $b$. We will later discuss how to remove this dependence in the reduction from weighted random-order streams. 

\begin{defn}
    Say that a batch is \emph{underfull} if it contains at least one underfull edge. Otherwise, if it does not contain any underfull edges, say that it is non-underfull.
\end{defn} 
We terminate Phase 1 when $\approx \poly(\epsilon) \frac{q}{b n}$ non-underfull batches arrive in a row.  
This ensures that Phase 1 terminates within the first $\approx n \beta^2 \cdot  \poly(\epsilon) \frac{q}{b n} \approx \frac{\epsilon}{b} q$ batches. Since each batch can contain at most $b$ edges, and since the arrival order of the batches is uniformly at random, it follows from Chernoff bounds that $$\mu(G_{late}) \geq  (1-2\epsilon )\mu(G).$$ Combining with Equation \ref{eqn:Bguarantee} we obtain $$\mu(H \cup U) \geq \left(\frac{2}{3}-3\epsilon\right) \mu(G).$$ 

 The only way for the space usage to become too large is if $\approx \poly(\epsilon) \frac{q}{b \cdot n}$ non-underfull batches arrive in a row when there are more than $C \cdot n \log n \poly(\frac{b}{\epsilon})$ underfull batches left in the stream, for some constant $C$. The probability of this event can be upper-bounded by 
\begin{equation*}
	\left(1- C \cdot  \frac{ n \log n}{q} \poly \left(\frac{b}{\epsilon}\right)\right)^{\poly(\epsilon/b)q/n} \leq n^{-C},
\end{equation*}
so with high probability, the algorithm stores at most $O(n \log n \poly(b))$ edges. 

In our reduction, the parameter $b$ corresponds to the maximum edge weight $W$ in the graph. This means that we would incur a polynomial dependence on $W$ in the space usage. However, a reduction due to Gupta and Peng \cite{GP} allows us to offset this space dependence. They devised a scheme for bucketing together edges according to their weight, which gives a reduction from general (possibly non-integral) weights, to integral bounded weights. Combining with this reduction, our algorithm uses space $O(n \log n \log R)$, where $R$ is the ratio between the heaviest and the lightest edge in the graph, and it can handle any (possibly non-integral) edge weights. In particular, the space usage is $O(n \polylog n)$ as long as the weights are polynomial in $n$. 

\subsection{Non-Bipartite Graphs}
In general, the reduction of Bernstein, Dudeja and Langley only holds for bipartite graphs. For non-bipartite graphs, it is no longer true that refolding preserves the matching size, since refolding a matching in $\phi(G)$ could incur an additional $2/3$ loss in the approximation ratio. Indeed, consider for example a weighted triangle with all edges of weight $2$ (see Figure~\ref{fig:nonbiprefolding}). 

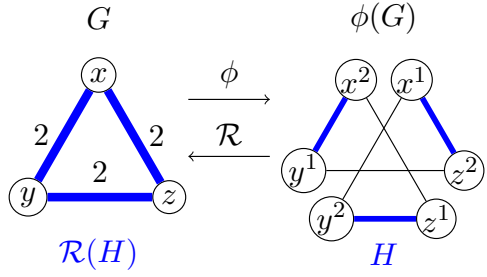
\begin{figure}[H]
\centering
\resizebox{0.4\linewidth}{!}{%
\begin{tikzpicture}
\draw[->] (1.1,0.7)-- (2.1,0.7) node[midway, above]{$\phi$};
\draw[->] (2.1,0)-- (1.1,0) node[midway, above]{$\mathcal{R}$};
\node (G) at (0,1.7) {$G$};
\begin{scope}[every node/.style={draw, circle, fill=white, inner sep=1.5pt}]
\node (x) at (90:1) {$x$}; 
\node (y) at (210:1) {$y$}; 
\node (z) at (-30:1) {$z$}; 
\end{scope}
\draw[blue, line width = 3pt] (x) -- (y) node[midway, left,black]{$2$};
\draw[blue, line width = 3pt]  (y) -- (z) node[midway, above, black]{$2$}; 
\draw[blue, line width = 3pt]  (z) -- (x) node[midway, right, black]{$2$};
\node[blue] (R text) at (0,-1.2){$\mathcal{R}(H)$};
\node (phiG) at (3.5,1.7) {$\phi(G)$};
\begin{scope}[every node/.style={draw, circle, fill=white, inner sep=0pt,minimum size=1pt},shift={(0:3.5)}]]
\node (x1) at (70:1) {$x^1$}; 
\node (x2) at (110:1) {$x^2$}; 
\node (y1) at (190:1) {$y^1$}; 
\node (y2) at (230:1) {$y^2$}; 
\node (z1) at (-50:1) {$z^1$}; 
\node (z2) at (-10:1) {$z^2$}; 
\end{scope}
\draw[blue, line width = 2pt] (x2) -- (y1) node[midway, left,black]{};
\draw[blue, line width = 2pt]  (y2) -- (z1) node[midway, above, black]{};
\draw[blue, line width = 2pt]  (z2) -- (x1) node[midway, right, black]{};
\draw (x1) -- (y2) node[midway, left,black]{};
\draw(y1) -- (z2) node[midway, above, black]{};
\draw (z1) -- (x2) node[midway, right, black]{};
\node[blue] (R text) at (3.5,-1.2){$H$};
\end{tikzpicture}
}
\caption{Refolding does not in general preserve matching size in non-bipartite graphs. Consider for example the blue subgraph $H = \{(x^1,z^2), (z^1,y^2), (y^1,x^2) \} \subseteq \phi(G)$ shown in the diagram. Then $\mu(H) = 3$, but $\mu_w(\mathcal{R}(H)) =2.$}
\label{fig:nonbiprefolding}
\end{figure}

We prove that the subgraph $H \cup U$ computed by Bernstein's algorithm still satisfies $\mu_w(\mathcal{R}(H \cup U )) \geq (2/3-\epsilon) \mu_w(G)$, even for non-bipartite graphs. This allows us to apply the unfolding reduction without any loss in the approximation ratio.  We achieve this by reducing to the bipartite case: We show that for every weighted graph $G$, there exists a bipartite subgraph $\widetilde{G} \subseteq G$ such that $\mu((H \cup U) \cap \phi(\widetilde{G})) \geq (2/3-\epsilon) \mu_w(G)$. We can then apply Lemma \ref{lemma:matchsize} to the bipartite graph $\widetilde{G}$ to get the result. 

In order to ``bipartify" the graph, we use the following lemma from \cite{BDL}, which says that there exists a bipartite subgraph in which the degrees to $H$ concentrate well (See Lemma \ref{lemma:bernstein4.2} for the formal statement). 
\begin{lemma}[Informal version of Lemma 5.7 in \cite{BDL}]\label{lemma:bipartifysimple}
Let $G$ be a weighted graph and let $M^*$ be a maximum weight matching in $G$. Suppose that $H \subseteq \phi(G)$ satisfies Equation \ref{eqn:underfull}. Then there exists a bipartite subgraph $\widetilde{G} \subseteq G$ such that $\widetilde{G}$ contains $M^*$, and, setting $\widetilde{H}:= H \cap \widetilde{G}$, it holds that $$\deg_{\widetilde{H}}(v) \approx \frac{\deg_H(v)}{2} \qquad \forall v \in V.$$
\end{lemma}
To show that $(H \cup U)\cap \phi(\widetilde{G})$ contains a large matching, we will in fact show that it is an EDCS:  
\begin{defn}[EDCS \cite{BS_EDCSbip}]\label{def:EDCS} Let $G = (V, E)$ be an unweighted graph, and $H = (V, E_H )$ a subgraph of $G$. Given parameters $\beta \geq 2$ and $\lambda < 1$, we say that $H$ is a $(\beta, \lambda)$-EDCS of $G$ if $H$ satisfies the following properties:
	\begin{itemize}
		\item (Property P1:) For all edges $(u, v) \in H$, it holds that $\deg_H(u) + \deg_H(v) \leq \beta.$
		\item (Property P2:) For all edges $(u, v) \in G \setminus H$, it holds that $\deg_H(u) + \deg_H(v) \geq \beta (1 - \lambda)$.
	\end{itemize}
\end{defn}
 The crucial property of EDCS is that it contains a $2/3$-approximate maximum cardinality matching. This was first proved in \cite{BS_EDCSbip} for bipartite graphs and in \cite{BS16} for general graphs. See also Lemma 3.2 in \cite{AB} for a simpler proof with improved parameters. 
\begin{thm}[EDCS contain a $2/3$-approximate matching \cite{AB}]
\label{thm:EDCS}
Let $G$ be an unweighted graph and let $\epsilon < 1/2$ be a parameter. Let $\lambda, \beta$ be parameters with $\lambda \leq \frac{\epsilon}{64}, \beta \geq 8\lambda^{-2} \log (1/\lambda)$. Then, for any $(\beta, \lambda)$-EDCS $H$ of $G$, we have that $\mu(H) \geq (\frac{2}{3}-\epsilon) \mu(G).$ 
\end{thm}

Now consider the weighted input graph $G$. Fix a maximum weight matching $M^*$ in $G$ and let $H$ be the graph computed by Phase 1 of Bernstein's algorithm on input $\phi(G)$.  Let $\widetilde{G} \subseteq G$ be the bipartite subgraph from Lemma \ref{lemma:bipartifysimple}. Ideally, we would like to show that $(H \cup U) \cap \phi(\widetilde{G})$ is an EDCS. However, this is not true in general, since the degrees to $U$ can be arbitrarily large (consider for example the case when $U$ is a star, see Figure \ref{fig:edcsrefolding} for an illustration), so $\deg_{(H \cup U) \cap \phi(\widetilde{G})}$ cannot be upper-bounded by a constant. Instead, we will sparsify $U$, so that its contribution to the degrees becomes insignificant. Let $\widetilde{H} = H \cap \phi(\widetilde{G})$ and let $\widetilde{U} = U \cap \phi(M^*)$ (see Figure \ref{fig:edcsrefolding}). This idea is similar to Bernstein's original analysis \cite{Bernstein}, except that now we perform this sparsification in the unfolded and ``bipartified" graph. 

\begin{figure}[H]
\centering
\resizebox{\linewidth}{!}{%
\usetikzlibrary{arrows}

\begin{tikzpicture}
\node[align = center] at (0,11.2) {Weighted}; 
\node[align = center] at (-3.5 ,9) {Non-bipartite}; 
\node (G) at (0,10.6) {$G$};
\node (text_v) at (-1,10.0) {\footnotesize High-degree vertex affecting the edge-degrees};
\draw [->] (text_v) to [bend right=25] (-0.3,9.5); 
\begin{scope}[every node/.style={draw, circle, fill=black,inner sep=1pt}]
\node (a1) at (-0.2,9.4){};
\node (a2) at (0.6,9.4){}; 
\node (a3) at (1.4,9.4){}; 
\node (b1) at (-0.2,8.6){};
\node (b2) at (0.6,8.6){};
\node (b3) at (1.4,8.6){}; 
\node (b4) at (-0.6,8.6){};
\node (b5) at (-1.0,8.6){};
\node (b6) at (-1.4,8.6){}; 
\end{scope}
\draw[black, line width = 1pt, ] (a1) -- (a2); 
\draw[black, line width = 1pt, ] (a2) -- (a3); 
\draw[black, line width = 1pt, ] (b1) -- (b2); 
\draw[black, line width = 1pt, ] (b2) -- (b3); 
\draw[black, line width = 1pt, ] (a1) -- (b2); 
\draw[black, line width = 1pt, ] (a1) -- (b3); 
\draw[black, line width = 1pt,] (a1) -- (b4); 
\draw[black, line width = 1pt, ] (a1) -- (b5); 
\draw[black, line width = 1pt, ] (a1) -- (b6);
\draw[red, line width = 1pt,] (a1) -- (b1); 
\draw[red, line width = 1pt, ] (a2) -- (b2); 
\draw[red, line width = 1pt, ] (a3) -- (b3);
\node[red] (M) at (2,9) {$M^*$}; 

\node[align = center] at (-3.5,5) {Bipartite}; 
\draw[->] (0,8)--(0,7)  node[midway, right, font = \small]{Bipartify};
\node (Gtilde) at (0,6.6) {$\tilde{G}$};
\begin{scope}[every node/.style={draw, circle, fill=black,inner sep=1pt}]
\node (c1) at (-0.2,5.4){};
\node (c2) at (0.6,5.4){}; 
\node (c3) at (1.4,5.4){}; 
\node (d1) at (-0.2,4.6){};
\node (d2) at (0.6,4.6){};
\node (d3) at (1.4,4.6){}; 
\node (d4) at (-0.6,4.6){};
\node (d5) at (-1.0,4.6){};
\node (d6) at (-1.4,4.6){}; 
\end{scope}
\draw[black, line width = 1pt, ] (c1) -- (d2); 
\draw[black, line width = 1pt, ] (c1) -- (d3); 
\draw[black, line width = 1pt,] (c1) -- (d4); 
\draw[black, line width = 1pt, ] (c1) -- (d5); 
\draw[black, line width = 1pt, ] (c1) -- (d6);
\draw[red, line width = 1pt,] (c1) -- (d1); 
\draw[red, line width = 1pt, ] (c2) -- (d2); 
\draw[red, line width = 1pt, ] (c3) -- (d3);
\node[red] (M) at (2,5) {$M^*$}; 
\node[align = center] at (8.5,11.2) {Unweighted}; 
\draw[->] (3,9)--(5.5,9); 
\node[font=\small, align = center] (text) at (4.25, 9.7) {Unfold and run \\ Bernstein's algorithm};
\node (phiG) at (8.5,10.6) {$\phi(G)$};
\node (text_u) at (8.5,10.0) {\footnotesize Underfull edges are marked blue};
\draw [->] (text_u) to [bend right=25] (7.8,9.4); 
\begin{scope}[every node/.style={draw, circle, fill=black,inner sep=1pt}]
\node (e1) at (8.3,9.4){};
\node (e2) at (9.1,9.4){}; 
\node (e3) at (9.9,9.4){}; 
\node (f1) at (8.3,8.6){};
\node (f2) at (9.1,8.6){};
\node (f3) at (9.9,8.6){}; 
\node (f4) at (7.9,8.6){};
\node (f5) at (7.5,8.6){};
\node (f6) at (7.1,8.6){}; 
\end{scope}
\draw[blue!30, line width = 5pt,-round cap] (e1)--(f1); 
\draw[blue!30,line width = 5pt, -round cap] (e1)--(f6); 
\draw[blue!30,line width = 5pt, -round cap] (e1)--(f5); 
\draw[blue!30,line width = 5pt, -round cap] (e1)--(f4); 
\draw[black, line width = 1pt, -round cap] (e1) -- (e2); 
\draw[black, line width = 1pt] (e2) -- (e3); 
\draw[black, line width = 1pt, ] (f1) -- (f2); 
\draw[black, line width = 1pt, ] (f2) -- (f3); 
\draw[black, line width = 1pt, ] (e1) -- (f2); 
\draw[black, line width = 1pt, ] (e1) -- (f3); 
\draw[black, line width = 1pt,]  (e1) -- (f4); 
\draw[black, line width = 1pt, ] (e1) -- (f5); 
\draw[black, line width = 1pt, ] (e1) -- (f6);
\draw[red, line width = 1pt,]   (e1) -- (f1); 
\draw[red, line width = 1pt, ] (e2) -- (f2); 
\draw[red, line width = 1pt, ] (e3) -- (f3);
\node[red]  at (10.7,9) {$\phi(M^*)$}; 
\node[blue] (U) at (7.7,8.2) {$U$};

\draw[->] (8,7.8)--(8,7); 
\node[font=\small, align = center] (text2) at (9.5,7.5){Bipartify and \\ Sparsify};
\draw[->] (3,5)--(5.5,5); 
\node[font=\small, align = center] (text3) at (4.25, 5.7) {Unfold and\\ Sparsify};
\node at (8.5,6.6) {$\tilde{H} \cup \phi(M^*) \subseteq \phi(\tilde{G})$};
\node (text_res) at (8.5,6) {\footnotesize Restrict $U$ in order to reduce edge-degrees};
\draw [->] (text_res) to [bend right=25] (7.8,5.4); 
\begin{scope}[every node/.style={draw, circle, fill=black,inner sep=1pt}]
\node (g1) at (8.3,5.4){};
\node (g2) at (9.1,5.4){}; 
\node (g3) at (9.9,5.4){}; 
\node (h1) at (8.3,4.6){};
\node (h2) at (9.1,4.6){};
\node (h3) at (9.9,4.6){}; 
\node (h4) at (7.9,4.6){};
\node (h5) at (7.5,4.6){};
\node (h6) at (7.1,4.6){}; 
\end{scope}
\draw[blue!30,line width = 5pt,-round cap] (g1)--(h1); 
\draw[black, line width = 1pt, ] (g1) -- (h2); 
\draw[black, line width = 1pt, ] (g1) -- (h3); 
\draw[red, line width = 1pt,] (g1) -- (h1); 
\draw[red, line width = 1pt, ] (g2) -- (h2); 
\draw[red, line width = 1pt, ] (g3) -- (h3);
\node[red] at (10.7,5) {$\phi(M^*)$}; 
\node[blue] (U) at (8.5,4) {$\tilde{U} = U \cap \phi(M^*)$}; 
\end{tikzpicture}
}
\caption{Illustration of the reduction to the bipartite case. We show that $\widetilde{H} \cup \widetilde{U}$ contains a matching of size at least $\left(\frac{2}{3}-\epsilon \right) \mu_w(G)$. Since $\widetilde{G}$ is bipartite, we can refold $\widetilde{H} \cup \widetilde{U}$ without reducing the matching size.}
\label{fig:edcsrefolding}
\end{figure}
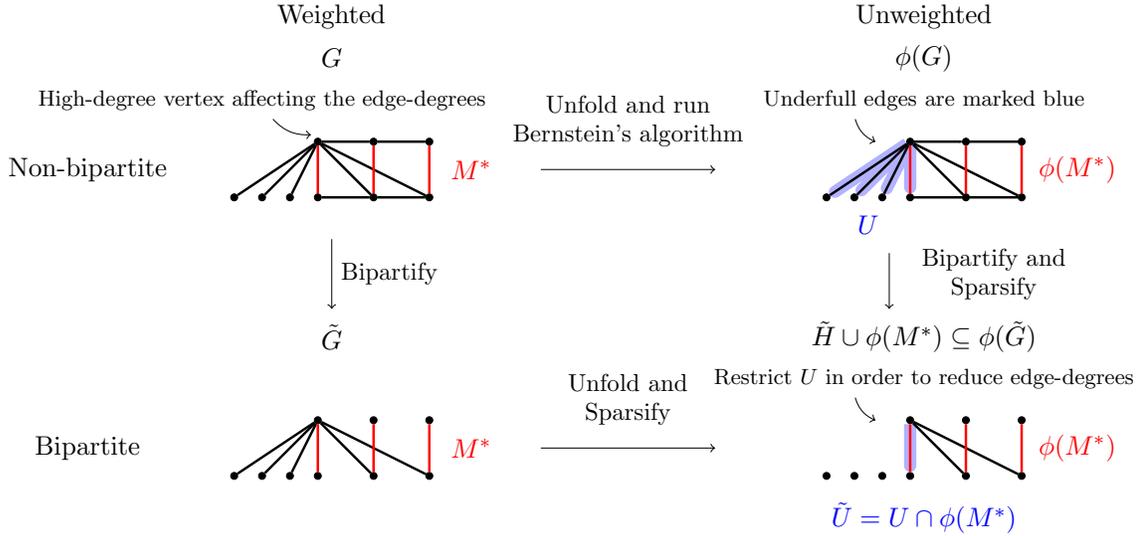

Now $\widetilde{U}$ is a matching, so for all $v \in V$, we have $\deg_{\widetilde{H} \cup \widetilde{U}}(v) \in \{ \deg_{\widetilde{H}}(v), \deg_{\widetilde{H}}(v)+1\}$. So
$$\deg_{\widetilde{H} \cup \widetilde{U}}(v) \approx \deg_{\widetilde{H}}(v) \approx \frac{1}{2} \deg_{H}(v).$$
In particular, 
$$\forall (u,v) \in \widetilde{H} \cup \widetilde{U}, \qquad \deg_{\widetilde{H} \cup \widetilde{U}}(u) + \deg_{\widetilde{H} \cup \widetilde{U}}(v) \approx \frac{1}{2} \deg_{H}(u) + \frac{1}{2} \deg_{H}(v) \leq \frac{\beta}{2},$$
and 
$$\forall (u,v) \in \phi(M^*) \setminus (\widetilde{H} \cup \widetilde{U}), \qquad \deg_{\widetilde{H} \cup \widetilde{U}}(u) + \deg_{\widetilde{H} \cup \widetilde{U}}(v) \approx \frac{1}{2} \deg_{H}(u) + \frac{1}{2} \deg_{H}(v) \geq \frac{\beta}{2} -1 .$$
Setting  $X= \widetilde{H} \cup \widetilde{U}$, $\beta' \approx \frac{\beta}{2}$, and $\lambda'$ to be a sufficiently small constant, we can now apply Theorem \ref{thm:EDCS}  to the graph  $\widetilde{H} \cup \phi(M^*)$, and obtain
$$ \mu(\widetilde{H} \cup \widetilde{U}) \geq \left(2/3-\epsilon \right) \mu(\widetilde{H} \cup \phi(M^*)) \geq  \left(2/3-\epsilon \right) \mu(\phi(M^*)). $$
Since $\widetilde{H} \cup\widetilde{U} \subseteq \phi(\widetilde{G})$ and since $\widetilde{G}$ is bipartite, we can apply Lemma \ref{lemma:matchsize} to get the required result
$$\mu_w(\mathcal{R}(H \cup U))\geq \mu_w(\mathcal{R}(\widetilde{H} \cup\widetilde{U})) = \mu(\widetilde{H} \cup\widetilde{U}) \geq   \left(2/3-\epsilon \right) \mu(\phi(M^*)) =  \left(2/3-\epsilon \right) \mu_w(G). $$

In the rest of the paper, we will present the full analysis. In Section \ref{sec:stream}, we formally present the algorithm and analysis for random-order streams. In Section \ref{sec:Comm}, we prove Theorem \ref{thm:main_2party} and Theorem \ref{thm:main_kparty}. 

\section{Notation and Preliminaries}

Given a graph $G = (V,E)$, we will use $n:=|V|$ to denote the number of vertices and $m:=|E|$ to denote the number of edges in $G$. If $G$ is weighted, then we will use $w: E \rightarrow \mathbb{R}^+$ to denote the edge weights, and $R:= \max_{e \in E} w_e/\min_{e \in E} w_e $ to denote the ratio between the heaviest and the lightest edge in $G$. We use $\mu(G)$ to denote the size of the maximum cardinality matching in $G$, and $\mu_w(G)$ to denote the weight of the maximum weight matching in $G$.

Given $\epsilon > 0$, let $\gamma_\epsilon := (1/\epsilon)^{\Theta(1/\epsilon)}$ be a large constant which will be incurred in the space usage of our algorithms (instead of a dependence on the maximum weight of the graph). Note that for any fixed $\epsilon$, we have $\gamma_{\epsilon} = O(1)$. 

\subsection{Models}
\begin{description}
    \item[Random-order streams] In the random-order stream model, the weighted edges of the input graph arrive one by one in an order chosen uniformly at random from all possible orderings. The algorithm makes a single pass over the stream and must output an approximate maximum weight matching at the end of the stream. 

    \item[Robust communication model] In the $k$-party one-way robust communication model, each weighted edge of the input graph is assigned independently and uniformly at random to one of the $k$ parties. The $i$th party is supplied with its assigned edges and a message $m_{i-1}$ from the $(i-1)$st party, and must send a message $m_i$ to the $(i+1)$st party. The $k$th party must output a valid weighted matching of the input graph. The communication complexity of a protocol is defined to be $\max_{1 \leq i \leq k} |m_i|$, where $|m_i|$ is the number of words in message $m_i$.\\
    In the case when $k=2$, we refer to the first party as Alice and to the second party as Bob. 
\end{description}
\subsection{Graph Unfolding}
In addition to the facts already stated in Section \ref{sec:tech}, we will need the following: 
\begin{thm}[Unfolding preserves matching size in bipartite graphs \cite{KLST}]\label{thm:refolding}
If G is a weighted bipartite graph, then $\mu_w(G) = \mu(\phi(G))$.
\end{thm}
\begin{defn}[Refolding approximate \cite{BDL}]\label{def:alphaapprox}
Let $G$ be a weighted graph. A subgraph $H \subseteq \phi(G)$ is $\alpha$-refolding-approximate if
$\mu_w(\mathcal{R}(H)) \geq \alpha \cdot \mu_w(G).$
\end{defn}

\subsection{EDCS}
We will use the following guarantee which holds for a certain relaxed notion of EDCS. 
\begin{defn}[Bounded edge-degree \cite{Bernstein}]\label{def:bdddeg}
    	We say that a graph $H$ has bounded edge-degree $\beta$ if for every edge $(u, v) \in H$, it holds that $\deg_H(u) + \deg_H(v) \leq \beta$.
\end{defn}
\begin{defn}[Underfull edge \cite{Bernstein}]\label{def:underfull}
    Let $G$ be any unweighted graph, and let $H$ be a subgraph of $G$ with bounded edge-degree $\beta$. Given a parameter $\lambda < 1$, we say that an edge $(u, v) \in G\setminus H$ is $(G, H, \beta, \lambda)$-underfull if $\deg_H(u) + \deg_H(v) < \beta(1 - \lambda)$.
\end{defn}
\begin{lemma}[Relaxed EDCS contain a $2/3$-approximate matching \cite{Bernstein}]\label{lemma:bernstein4.1}
Let $\epsilon < \frac{1}{2}$ be a parameter, and let $\lambda, \beta$ be parameters with $\lambda \leq \frac{\epsilon}{128}$, $\beta \geq 16\lambda^{-2} \log(1/\lambda)$. Consider any unweighted graph $G$ and any subgraph $H$ with bounded edge-degree $\beta$. Let $U$ contain all edges in $G \setminus H$ that are $(G,H,\beta, \lambda)$-underfull. Then $\mu(H \cup U) \geq (2/3 - \epsilon)\mu(G).$
\end{lemma}
\subsection{Approximate Blossom Inequality}
The following fact will allow us to round fractional weighted matchings. We refer readers to Lemma 2.2 in \cite{blossom} for a proof sketch, and to Section 25.2 of \cite{Schrijver} for a detailed discussion on blossom inequalities. 
\begin{proposition}[Folklore]\label{prop:blossom}
	Let G be a weighted graph, and let $x$ be a fractional matching on $G$. Suppose that for every set $S \subseteq V$ with $|S| \leq \frac{1}{\epsilon}$, we have
	$$\sum_{e \in  G[S]} x_e  \leq  \left\lfloor\frac{|S|}{2} \right \rfloor.$$ Then $\mu_w(G) \geq (1-\epsilon)\sum_e w_e x_e.$
\end{proposition}

\subsection{Concentration Inequality}
We will use the Chernoff bound for negatively associated random variables (see e.g. the primer in \cite{NA}).

\begin{thm}\label{thm:chernoff_neg}
    Let $X_1, \dots X_n$ be negatively associated random variables taking values in $[0,1].$ Let $X:= \sum X_i$ and let $\mu:= \mathbb{E}[X]$. Then, for any $0 < \delta < 1$, we have 
    $$\Pr[X \leq \mu (1-\delta)] \leq \exp{\Bigr(\frac{- \mu \delta^2}{2}\Bigr)}.$$
\end{thm}

\section{$2/3$-Approximation in Random-Order Streams}\label{sec:stream}
In this section, we prove Theorem \ref{thm:main_stream}. In Section \ref{sec:reduction}, we formally describe the reduction from weighted random-order streams to unweighted $b$-batch random-order streams, and we prove its correctness. In Section \ref{section:bip}, we show that Bernstein's $2/3$-approximation algorithm \cite{Bernstein} for random-order streams still works under batch arrivals. Finally, in Section \ref{sec:nonbip}, we show that the obtained weighted random-order streaming algorithm still works for non-bipartite graphs, and we complete the proof of Theorem \ref{thm:main_stream}. 

\subsection{Reduction to Unweighted $b$-Batch Random-Order Streams}\label{sec:reduction}
Gupta and Peng \cite{GP} gave a reduction which allows us to assume that the edge weights are integral and bounded above by a large constant. They originally proved the reduction for the dynamic graph model, however, it also applies to the streaming and one-way communication models (See Theorem 6.1 and Theorem 6.2 in \cite{BDL}). 
\begin{thm} [Reduction to bounded integral weights \cite{GP, BDL}] \label{thm:space} If there is a random-order streaming algorithm $\mathcal{A}$ to compute an $\alpha$-approximate maximum weight matching in graphs with weights in $[W]$ using space $S(n, m, W, \alpha)$, then there exists a random-order streaming algorithm $\mathcal{A}'$ to compute a $(1-\epsilon)\alpha$-approximate maximum weight matching in graphs with weights in $\mathbb{R}^{+}$ using space $O(S(n, m, \gamma_{\epsilon}, \alpha) \log R)$. 

Similarly, if there is a one-way robust communication complexity protocol to compute an $\alpha$-approximate maximum weight matching in graphs with weights in $[W]$ using $C(n,m,W, \alpha)$ words of communication, then there exists a protocol to compute a $(1-\epsilon) \alpha$-approximate maximum weight matching in graphs with weights in $\mathbb{R}^{+}$ using $O(C(n, m, \gamma_{\epsilon}, \alpha) \log R)$ words of communication. 
\end{thm}
We would like to use the unfolding technique to reduce to the unweighted problem. As Bernstein, Dudeja and Langley \cite{BDL} showed, in \emph{adversarially ordered} streams, unfolding immediately gives a reduction for bipartite graphs: Whenever a weighted edge $e \in G$ arrives, we can unfold it and pass the unweighted edges $\phi(e)$ sequentially into an unweighted streaming algorithm while tracking the updates in the weighted stream. In \emph{random-order} streams, there is a subtle issue with this approach. If the edges arrive uniformly at random in $G$, then the arrival order in $\phi(G)$ will not be uniformly at random, but rather there will be batches of edges that necessarily arrive together. To overcome this issue, we consider the $b$-batch random-order stream model, restated below. 
\bbatch*

Graph unfolding naturally gives a reduction from weighted random-order streams to unweighted $b$-batch random-order streams. 

\begin{thm}[Reduction to the $b$-batch model]\label{thm:reduction}
 If there exists an algorithm $\mathcal{A}_B$ for the unweighted $b$-batch random-order stream model that computes an $\alpha$-approximate maximum cardinality matching in bipartite graphs using space $S(n,m,b,\alpha)$, then there exists an algorithm $\mathcal{A}_w$ for weighted random-order streams (with weights in $\mathbb{R}^+$) that computes a $(1-\epsilon) \alpha$-approximate maximum weight matching in bipartite graphs using space $O(S(n \gamma_{\epsilon},m \gamma_{\epsilon},\gamma_{\epsilon},\alpha) \log R).$ 
    
    Moreover, suppose that $\mathcal{A}_B$ computes an $\alpha$-refolding approximate subgraph whenever the input graph is of the form $\phi(G)$ for some weighted graph $G$ with batches $\mathcal{B}=\{ \phi(e) : e \in G\}$. Then the guarantees of $\mathcal{A}_w$ also hold for non-bipartite graphs.   
\end{thm}
\begin{proof}
    Let $\mathcal{A}_B$ be the unweighted $b$-batch random-order streaming algorithm using space $S(n,m,b,\alpha)$. By Theorem \ref{thm:space}, it suffices to construct an algorithm $\mathcal{A}_W$ that computes an $\alpha$-approximate maximum weight matching using space $S(Wn, Wm, W, \alpha)$ when the edge weights are in $[W]$.  We can obtain the required algorithm $\mathcal{A}_W$ as follows: 

    Whenever an edge $e$ arrives in the weighted stream, define a batch $B_e := \phi(e)$ consisting of the unfolded edges of $e$. Feed the batch $B_e$ as an update to the batch algorithm $\mathcal{A}_B$. In other words, $\mathcal{A}_B$ is applied to the graph $\phi(G)$ with batches $\mathcal{B} = \{ \phi(e): e\in G\}$. At the end of the stream, $\mathcal{A}_B$ outputs an $\alpha$-approximate maximum cardinality matching $M$ of $\phi(G)$. The algorithm $\mathcal{A}_W$ outputs the maximum weight matching in $\mathcal{R}(M)$ (which can easily be computed from $M$). We have
    \begin{align*}
    \mu_w(\mathcal{R}(M)) &  \geq \mu(M), & \text{by Lemma \ref{lemma:matchsize}} \\
    & \geq \alpha \cdot \mu(\phi(G)), & \text{by the assumption on $M$}\\
    &= \alpha \cdot \mu_w(G), & \text{by Theorem \ref{thm:refolding} }
    \end{align*}
    so $\mathcal{A}_W$ outputs an $\alpha$-approximate maximum weight matching. Since the graph $\phi(G)$ has at most $Wn$ vertices and $Wm$ edges, the space usage of $\mathcal{A}_W$ is at most $S(Wn, Wm, W, \alpha)$, as required. 

    For the ``Moreover''-part, suppose that $\mathcal{A}_B$ computes an $\alpha$-refolding approximate subgraph $H$. Define $\mathcal{A}_W$ as before, except that now $\mathcal{A}_W$ should output the maximum weight matching in $\mathcal{R}(H)$. Then 
    \begin{align*}
    \mu_w(\mathcal{R}(H)) \geq \alpha \cdot \mu_w(G), \qquad \text{by Definition \ref{def:alphaapprox}},
    \end{align*}
    so the approximation ratio achieved by $\mathcal{A}_W$ is still $\alpha$, even for non-bipartite graphs, as required. 
    \end{proof}
\begin{remark}
The argument can easily be adapted to the robust communication model. Consider a $b$-batch robust communication model, in which an adversary partitions the edges into batches of size at most $b$, and each batch gets assigned uniformly at random to each of the parties. Then any protocol for the $b$-batch robust communication model gives a protocol for the weighted robust communication model. 
\end{remark}
\subsection{$2/3$-Approximation in $b$-Batch Random-Order Streams}\label{section:bip}
 In this section, we prove the following proposition. 
 \begin{proposition}\label{prop:batch_algo}
    Given any unweighted graph $G$ and any approximation parameter $0< \epsilon < 1$, Bernstein's algorithm (Algorithm \ref{alg:batch_bernstein}) with high probability computes a $(2/3 -\epsilon)$- approximate maximum cardinality matching in the $b$-batch random-order stream model. The space complexity of the algorithm is $O(n b^2  \log n \log b\poly(\epsilon^{-1}))$, where $b$ is the upper bound on batch size. 
\end{proposition}

\begin{defn}[Parameters] Let $\epsilon < \frac{1}{2}$ be the final approximation parameter we are aiming for, $\lambda = \frac{\epsilon}{512}, \beta = 144\lambda^{-2}\log(2b/\lambda)$; note that $\beta = O(\poly(\epsilon^{-1})\log b)$. Set $\alpha =\frac{\epsilon q}{b(n\beta^2+1)}$ and $\gamma = 7 \log n \frac{q}{\alpha} = O(nb \log n\log b\poly(\epsilon^{-1}))$. 
\end{defn}
We now describe Bernstein's algorithm \cite{Bernstein} adapted to the $b$-batch model. The algorithm has two Phases. In Phase 1, it computes a subgraph $H$ that is bounded edge-degree $\beta$ (Definition \ref{def:bdddeg}). In Phase 2, it stores all the $(G,H,\beta,\lambda)$-underfull edges (Definition \ref{def:underfull}). That way, the algorithm computes a `relaxed' EDCS, which by Lemma \ref{lemma:bernstein4.1} contains a $(2/3-\epsilon)$-approximate maximum cardinality matching. 

More concretely, the algorithm proceeds as follows: 
Initialize an empty subgraph $H$ and start Phase 1. Phase 1 is split into epochs, each of which contains exactly $\alpha$ batches. In each epoch, the algorithm processes the batches sequentially. For each edge $(u,v)$ in the batch, if $\deg_H(u) + \deg_H(v) <\beta(1-\lambda)$, then $(u,v)$ is added to the subgraph $H$ (note that the algorithm changes $H$ over time, so $\deg_H$ always refers to the degree of $H$ at the time when the edge is examined). After adding an edge to $H$, the algorithm runs procedure \textsc{RemoveOverfullEdges}(H) to ensure that $H$ is always bounded edge-degree $\beta$. In each epoch, the algorithm also has a boolean \textsc{FoundUnderfull}, which tracks whether at least one underfull edge arrived in the epoch. If \textsc{FoundUnderfull} is FALSE at the end of an epoch, then the algorithm terminates Phase 1 and moves on to Phase 2. Once Phase 1 terminates, the subgraph $H$ becomes fixed and does not change anymore. Then, in Phase 2, the algorithm simply picks up all the underfull edges into a separate set $U$. 
For a formal description, see Algorithm \ref{alg:batch_bernstein}. Note that the only difference between Algorithm \ref{alg:batch_bernstein} and Bernstein's original algorithm (Algorithm 1 in \cite{Bernstein}) is that the length of each epoch is now determined by the number of batches, rather than the number of edges. 

\begin{algorithm}[H]
\caption{Bernstein's Algorithm \cite{Bernstein} adapted to the $b$-batch model}\label{alg:batch_bernstein}
    \SetAlgoLined\DontPrintSemicolon
    \SetAlgoNoEnd
  \SetKwFunction{algo}{algo}
  \SetKwFunction{proc}{proc}
  \SetKwProg{myproc}{Procedure}{}{end}
  \SetKwProg{Do}{Do until termination}{}{}
   $H \leftarrow \emptyset,U \leftarrow \emptyset $ \\
    \myproc{\textsc{Phase 1}}{
    \Do{}{
    { \textsc{FoundUnderfull} $\leftarrow$ FALSE} \\
    \For(\tcp*[f]{Each epoch has $\alpha$ batches}){$i = 1, \dots, \alpha$}{
    Let $B_i$ denote the $i^{th}$ batch in the epoch\\
    \For{$(u,v) \in B_i$ }{
    \If{$\deg_H(u) + \deg_H(v) < \beta(1-\lambda)$}{$H \leftarrow H \cup \{(u,v)\}$ \\
    {\textsc{FoundUnderfull} $\leftarrow$ TRUE} \\
    \textsc{RemoveUnderfullEdges(H)}
    }
    }
    }
    \If{{\textsc{ FoundUnderfull} = FALSE}}{Go to Phase 2}
    }
    }
    \myproc{\textsc{ RemoveOverfullEdges(H)}}{
    \While{there exists $(u,v) \in H$ such that $\deg_H(u) + \deg_H(v) > \beta$}{ Remove $(u,v)$ from $H$}
    }
     \myproc{\textsc{Phase 2}}{
     \ForEach{remaining edge $(u,v)$ in the stream}{
    \If{$\deg_H(u) + \deg_H(v) < \beta(1-\lambda)$}{$U \leftarrow U \cup \{(u,v)\}$}}
     \Return{a maximum cardinality matching in $H \cup U$}
     }
\end{algorithm}
    \begin{defn}
    Let $\mathcal{B}_{early}$ denote the first $\frac{\epsilon}{b} q$ batches in the stream and let $\mathcal{B}_{late}$ denote the remaining batches. Let $E_{late}:= \{e \in E: e \in B \text{ for some } B \in \mathcal{B}_{late} \}$ be the set of edges that arrive as part of the late batches. More generally, let $E_{>i}$ denote the set of edges that arrive after the first $i$ batches. 
   \end{defn} 
 First, we show that we don't lose too many matching edges in the early part of the stream. To this end, we need to assume that the maximum cardinality matching is sufficiently large.
 
 \begin{claim}
        	\label{claim:largematching}
    	We can assume that $\mu(G) \geq 20b^2 \log n\epsilon^{-2}$. 
    \end{claim}
    \begin{proof}
    	It is well known that every graph $G$ satisfies $m \leq 2n\mu(G)$. Hence, if $\mu(G) <  20b^2 \ \log n\epsilon^{-2}$, then $m =  O( n b^2 \log n \epsilon^{-2})$, so we can simply store all the edges. 
    \end{proof}

    \begin{lemma}\label{lemma:whp_stream}
        For $\epsilon < b/2$, it holds that $\Pr[\mu(E_{late}) \geq (1 - 2\epsilon)\mu(G) ] \geq 1-n^{-5}.$
    \end{lemma}
    \begin{proof}
    	Fix a maximum cardinality matching in $M^*$ in $G$, and let $B_1, \dots, B_k$ be the set of batches containing at least one matching edge from $M^*$. Each batch $B_i$ can contain at most $b$ edges from $M^*$, so it suffices to show that at most $\frac{2\epsilon \mu(G)}{b}$ of these batches arrive in the early part of the stream. 
    	For $1 \leq i \leq k$, let $X_i$ be the indicator variable that $B_i \in \mathcal{B}_{late}$, and let $X = \sum_{i =1 }^k X_i.$ We will show that with high probability, $X \geq k - \frac{2\epsilon \mu(G)}{b}$. For $1 \leq i \leq k$, we have $\mathbb{E}[X_i] = (1-\frac{\epsilon}{b})$, so $\mathbb{E}[X] = k(1-\frac{\epsilon}{b})$. 
     
     The $X_i$s are negatively associated since these variables correspond to sampling $(1-\epsilon)q$ batches uniformly at random without replacement, which is known to be negatively associated (see the primer \cite{NA}).  Applying Theorem \ref{thm:chernoff_neg} gives
    	\begin{align*}
    		\Pr\left[X \geq k - \frac{2 \epsilon \mu(G)}{b}\right]  &= 1-  	\Pr \Bigr[X -\mathbb{E}[X] < -   \Bigr(\frac{2 \epsilon \mu(G)}{b} -\frac{\epsilon k}{b}\Bigr)\Bigr]   \\
    		& \geq 1-\exp \Bigr( \frac{-\epsilon^2 (2 \mu-k)^2}{4b^2k} \Bigr) \\
    		& \geq 1 -n^{-5}.
    	\end{align*}
    Here the last inequality follows because $\mu(G) \geq k$ and $\mu(G) \geq 20 \log nb^2\epsilon^{-2}$, by Claim \ref{claim:largematching}. 
   \end{proof}
\begin{lemma}\label{lemma:batchstream}
Phase 1 satisfies the following properties: 
\begin{enumerate}
        \item Phase 1 terminates within the first $\frac{\epsilon q}{b}$ batches of the stream.
        \item Phase 1 constructs a subgraph $H \subseteq G$ with bounded edge-degree $\beta$. As a corollary, $H$ has at most $O(n \beta)$ edges. 
        \item When Phase 1 terminates after processing some batch $B_l$, with probability at least $1-n^{-5}$, the total number of $(E_{>l},H,\beta, \lambda)$-underfull edges in $E_{>l} \setminus H $ is at most $b\gamma$.
\end{enumerate}
\end{lemma}
The proof of Lemma \ref{lemma:batchstream} proceeds similarly to the proof of Lemma 4.1 in \cite{Bernstein}. We will use the following result from the original analysis. 

\begin{lemma}[\cite{Bernstein}]\label{lemma:bernstein4.2}
Fix any parameter $\beta > 2$. Let $H = (V,E_H)$ be a graph, with $E_H$ initially empty. Say that an adversary adds and removes edges from $H$ using an arbitrary sequence of two possible moves. 
\begin{itemize}
    \item (Deletion move) Remove an edge $(u,v)$ from H for which $\deg_H(u) + \deg_H(v) > \beta$. 
    \item (Insertion move) Add an edge $(u,v)$ to H for some pair $u,v \in V$ for which $\deg_H(u) + \deg_H(v) < \beta -1$. 
\end{itemize}
Then, after $n\beta^2$ moves, no legal move remains. 
\end{lemma}
\begin{proof}[Proof of Lemma \ref{lemma:batchstream}]

\emph{Property 1:} By Lemma \ref{lemma:bernstein4.2},  the algorithm can make at most $n \beta^2$ changes to $H$. Since each epoch that does not terminate Phase 1 must make at least one change to $H$, there can be at most $n\beta^2 +1$ epochs in Phase 1. So overall, Phase 1 goes through at most $\alpha (n \beta^2 +1) = \frac{\epsilon q}{b}$ batches in Phase 1. \\
\emph{Property 2:} Holds by construction of the algorithm, since the \textsc{RemoveOverfull} procedure ensures that $H$ is always bounded edge-degree $\beta$. \\
\emph{Property 3:}
Let $l$ be the last batch processed in Phase 1. We will say that a batch $B_j$ with $j > l$ is \textit{underfull} if it contains at least one $(E_{>l},H,\beta, \lambda)$-underfull edge. 
We will show that with probability at least $1-n^{-5}$, the number of underfull batches is at most $\gamma$. Since each underfull batch contains at most $b$ underfull edges, this will give the result. The intuition is as follows: Phase 1 terminates only if there is an epoch without a single underfull batch. Since the stream is random, this means that there are relatively few underfull batches left in the stream. More formally, for each epoch $0 \leq i \leq \frac{\epsilon q}{b}$, say that a batch is \emph{underfull} if it contains at least one $(G, H_i, \beta, \lambda)$-underfull edge, where $H_i$ is the subgraph $H$ constructed by the algorithm at the beginning of epoch $i$. Let $\mathcal{E}_i$ be the event that no underfull batches appear in epoch $i$, and let $\mathcal{F}_i$ be the event that there are more than $\gamma$ underfull batches left in the stream when epoch $i$ begins. Property 3 fails only if $\mathcal{E}_i \land \mathcal{F}_i$ happens for some $i$, so we need to upper bound $ \Pr[\cup_{i=1}^{\epsilon q/b} \mathcal{E}_i \land \mathcal{F}_i]$. Let $\mathcal{B}_i^r$ denote the set of batches that have not yet appeared at the beginning of epoch $i$ (r for remaining), let $\mathcal{B}_i^e$ denote the set of batches that appear in epoch $i$ (e for epoch) and let $\mathcal{B}_i^u$ denote the set of underfull batches that remain in the stream (u for underfull). With these definitions, we can write $\mathcal{E}_i \land \mathcal{F}_i$ as the event $ (\mathcal{B}_i^u \cap \mathcal{B}_i^e = \emptyset)\land (|\mathcal{B}_i^u| > \gamma)$, since the event $\mathcal{B}_i^u \cap \mathcal{B}_i^e = \emptyset$ ensures that the graph $H$ does not change throughout the epoch. We have  
\begin{align*}
    \Pr \left[\mathcal{E}_i \land \mathcal{F}_i \right] &=  \Pr\left[ (\mathcal{B}_i^u \cap \mathcal{B}_i^e = \emptyset)\land (|\mathcal{B}_i^u| > \gamma) \right] \\
    & \leq \Pr\left[\mathcal{B}_i^u \cap \mathcal{B}_i^e = \emptyset \bigr| |\mathcal{B}_i^u | > \gamma \right] \\
    &< \left(1-\frac{\gamma}{q}\right)^{\alpha} \\
    & = \left(1-\frac{7\log n}{\alpha}\right)^{\alpha} \\ 
    & \leq n^{-7}.
    \end{align*}
    Here the second inequality follows because $\mathcal{B}_i^e$ is obtained by sampling $\alpha$ batches from $\mathcal{B}_i^r$ uniformly at random without replacement, and since $ |\mathcal{B}_i^u| > \gamma$ and $|\mathcal{B}_i^r| \leq q$. There are at most $n^2$ epochs in total, so taking the union bound over all epochs gives the result.
    \end{proof}

Finally, we complete the proof of Proposition \ref{prop:batch_algo}. 

\begin{proof}[Proof of Proposition \ref{prop:batch_algo}]
Let us first show the approximation guarantee. By Part 2 of Proposition \ref{lemma:batchstream}, Phase 1 computes a subgraph $H$ which has 
bounded edge-degree $\beta$. Moreover, by Part 1 of Proposition \ref{lemma:batchstream}, it holds that $H \subseteq E_{late}$. Phase 2 finds the set $U$ of all $(E_{late}, H, \beta, \lambda)$-underfull edges. So by Lemma \ref{lemma:bernstein4.1} applied to the graph $E_{late}$, the algorithm returns a matching of size at least 
\begin{align*}
\mu(H \cup U ) & \geq (2/3-\epsilon) \mu(E_{late}) & \text{by Lemma \ref{lemma:bernstein4.1}} \\
& \geq (2/3 - \epsilon)(1-2\epsilon) \mu(G), & \text{by Lemma \ref{lemma:whp_stream}}, 
\end{align*}
where the last inequality holds with probability at least $1-n^{-5}.$ Re-scaling $\epsilon$ gives the approximation ratio. 

For the space analysis: By Part 2 of Lemma \ref{lemma:batchstream}, the space usage of Phase 1 is $O(n \beta) = O(n \log b\poly(\epsilon^{-1}))$. By Part 3 of Lemma \ref{lemma:batchstream}, with probability at least $1-n^{-5}$, the space usage of Phase 2 is at most $O(b \gamma) = O(n b^2 \log n   \log b \poly(\epsilon^{-1}))$. So with probability at least $1- n^{-5}$, the total space usage is at most $O(n b^2 \log n \log b \poly(\epsilon^{-1}))$. 
By a union bound, the overall success probability of the algorithm is at least $1-2n^{-5}$. 
\end{proof}
\subsection{Extension to Non-Bipartite Graphs}\label{sec:nonbip}
In this section, we show that the computed graph $H \cup U$ is $(2/3-\epsilon)$-refolding approximate. This, together with Proposition \ref{prop:batch_algo} and Theorem \ref{thm:reduction} will complete the proof of Theorem \ref{thm:main_stream}. 

In \cite{BDL}, it was shown that EDCS are (almost) $2/3$-refolding approximate. However, since $H \cup U$ is not an EDCS, but rather a relaxed version of an EDCS, this result cannot be applied directly. Instead, we need a more careful argument. First, we need the following lemma which was proved in \cite{BDL}. 
\pagebreak
\begin{lemma}[Lemma 5.7 in \cite{BDL}]
\label{lemma:bipartify}
Let $G$ be a weighted graph with weights in  $[W]$. Let $\delta \in (0, 1/2)$, and let $d \geq 36\delta^{-2} log(W /\delta)$. For any matching $M$ in $G$ and any subgraph $H$ of $\phi(G)$
with maximum degree at most $d$, there exists a bipartite subgraph $\widetilde{G} =\widetilde{G}(M, H) $ of $G$ such that, setting 
$\widetilde{H}:= H \cap \phi(\widetilde{G})$, it holds that
\begin{enumerate}
    \item $M \subseteq \widetilde{G}$ and
    \item $| \deg_{\widetilde{H}}(v) - \deg_{H}(v)/2| \leq \delta d$ for all vertices $v \in V(H)$. 
\end{enumerate}
\end{lemma}
\begin{remark}
Bernstein, Dudeja and Langley \cite{BDL} state the lemma for the special case when $M$ is a maximum weight matching in $G$, however, without changing their argument, the same is true for any arbitrary matching M.   
\end{remark}

\begin{lemma}\label{lemma:refolding_approx}
 Let $G$ be a (possibly non-bipartite) weighted graph with weights in $[W]$. Let  $\epsilon >0$, $\lambda \leq \frac{\epsilon}{512}$, $\beta  \geq \frac{144}{\lambda^2}\log(2W/\lambda)$. Let $G_S \subseteq G$ be a subgraph of $G$. Consider the unfolded graph $\phi(G)$. Let $H$ be a subgraph of $\phi(G)$ with bounded edge-degree $\beta$, and let $U$ be the set of all edges in $\phi(G_S) \setminus H$ that are $(\phi(G_S), H, \beta, \lambda)$-underfull. Then $\mu_w(\mathcal{R}(H \cup U)) \geq (2/3-\epsilon) \mu_w(G_S)$. 
\end{lemma}
\begin{proof}
Let $\delta = \frac{\lambda}{2}$. Fix a maximum-weight matching $M^*$ of $G_S$, and let $\widetilde{G}= \widetilde{G}(M^*, H)$ be the bipartite subgraph obtained from Lemma \ref{lemma:bipartify}. Consider the subgraph $\phi(\widetilde{G}) \subseteq \phi(G).$ Let $\widetilde{H}:=H \cap \phi(\widetilde{G})$ be the restriction of $H$ to $\phi(\widetilde{G})$, and let $\widetilde{U}:=U \cap \phi(M^*)$ be the restriction of $U$ to the matching $\phi(M^*)$. Note that $\widetilde{U}$ is a matching. By Lemma \ref{lemma:bipartify}, we have $\phi(M^*) \subseteq \phi(\widetilde{G})$, and therefore $\widetilde{H} \cup \widetilde{U} \subseteq \phi(\widetilde{G})$. Therefore, we may now apply Lemma \ref{lemma:matchsize} to the bipartite graph $\widetilde{G}$ and the subgraph $\widetilde{H} \cup \widetilde{U}$ of $\phi(\widetilde{G})$. 

\begin{gather}
\begin{aligned}\label{eqn:nonbipedcs1}
\mu_w(\mathcal{R}(H \cup U))  & \geq \mu_w(\mathcal{R}(\widetilde{H} \cup \widetilde{U})), \quad && \text{since $H \cup U \supseteq \widetilde{H} \cup \widetilde{U}$} \\
& \geq \mu(\widetilde{H} \cup \widetilde{U}), \quad  &&\text{by Lemma \ref{lemma:matchsize}}.
\end{aligned}
\end{gather}
Furthermore, recalling that $M^*$ is a maximum weight matching in $G_S$, we have
\begin{equation}\label{eqn:nonbipedcs2}
\mu_w(G_S)  = w(M^*) = |\phi(M^*)| \leq \mu(\widetilde{H} \cup \phi(M^*)).  
\end{equation}
We will show that $\mu(\widetilde{H} \cup \widetilde{U}) \geq (\frac23 - \epsilon) \mu(\widetilde{H} \cup \phi(M^*))$. To this end, we will show that $\widetilde{H} \cup \widetilde{U}$ is an EDCS of $\widetilde{H} \cup \phi(M^*)$. 
\begin{claim}\label{claim:relaxed_edcs}
$\widetilde{H} \cup \widetilde{U}$ is a $(\beta', \lambda')$-EDCS of $\widetilde{H} \cup \phi(M^*)$ for $\beta'=\frac{\beta}2 + \beta \lambda + 2, \lambda' = 8 \lambda$. 
\end{claim}
\begin{proof}
Let us start by showing property P1 in Definition \ref{def:EDCS}. First note that for all $(u,v) \in \widetilde{H} \cup \widetilde{U}$, it holds that  $\deg_H(u) + \deg_H(v) \leq \beta$. Indeed, if $(u,v) \in \widetilde{H}$, then $(u,v) \in H$, so $\deg_H(u) + \deg_H(v) \leq \beta$ since $H$ is bounded edge-degree $\beta$. If instead  $(u,v) \in \widetilde{U}$, then $(u,v) \in U$,  so $\deg_H(u) + \deg_H(v)\leq (1-\lambda) \beta$, since all elements of $U$ are $(\phi(G), H, \beta, \lambda)$-underfull. Therefore, for all $(u,v) \in \widetilde{H} \cup \widetilde{U}$, it holds that 
\begin{align*}
\deg_{\widetilde{H}\cup \widetilde{U}}(u) + \deg_{\widetilde{H}\cup \widetilde{U}}(v) & \leq \deg_{\widetilde{H}}(u) +  \deg_{\widetilde{H}}(v) + \deg_{\widetilde{U}}(u) + \deg_{\widetilde{U}}(v) & \\
& \leq \frac{\deg_H(u) + \deg_H(v)}{2}+2\delta \beta  + \deg_{\widetilde{U}}(u) + \deg_{\widetilde{U}}(v) , & \text{by Lemma \ref{lemma:bipartify}}\\
& \leq  \frac{\beta}{2}  +2\delta \beta  + \deg_{\widetilde{U}}(u) + \deg_{\widetilde{U}}(v) , & \text{since $\widetilde{U}$ is a matching} \\
& = \frac{\beta}{2} +  \beta \lambda   + \deg_{\widetilde{U}}(u) + \deg_{\widetilde{U}}(v),  & \text{since $\delta = \frac{\lambda}{2}$} \\
&\leq  \frac{\beta}{2} +  \beta \lambda + 2
\\
& =  \beta'. 
\end{align*}

We now show property P2 in Definition \ref{def:EDCS}: 
If $(u,v) \in (\widetilde{H} \cup \phi(M^*))\setminus (\widetilde{H} \cup \widetilde{U})$, then $(u,v) \in \phi(M^*) \setminus U$, and in particular $\deg_H(u) + \deg_H(v) >  (1-\lambda) \beta$ (by definition of $U$). Thus, 
\begin{align*}
\deg_{\widetilde{H}\cup \widetilde{U}}(u) + \deg_{\widetilde{H}\cup \widetilde{U}}(v) & \geq \deg_{\widetilde{H}}(u) + \deg_{\widetilde{H}}(v) && \\
& \geq \frac{\deg_H(u) + \deg_H(v)}{2}-2\delta \beta, && \text{by Lemma \ref{lemma:bipartify}} \\
& \geq \frac{\beta (1-\lambda)}{2}-2\delta \beta &&  \\
& =\frac{\beta (1-\lambda)}{2}- \beta \lambda, && \text{since $\delta = \frac{\lambda}{2}$} \\
& \geq \left(\frac{\beta}{2}+ \lambda \beta + 2\right)(1-8\lambda) &&  \\
&= \beta'(1-\lambda'). &&
\end{align*}
The last inequality follows from simple algebraic manipulations, using the fact that  $\lambda \beta \geq 2$. 
\end{proof}

By the choice of parameters, we have $\lambda' \leq \frac{\epsilon}{64}$ and $\beta' \geq 8 \lambda'^{-2}\log(1/\lambda')$, so Claim \ref{claim:relaxed_edcs} together with Theorem \ref{thm:EDCS} yields $\mu(\widetilde{H} \cup \widetilde{U}) \geq (2/3-\epsilon) \mu(\widetilde{H} \cup \phi(M^*))$. 
Combining everything, we get
\begin{align*}
    \mu_w(\mathcal{R}(H \cup U)) & \geq \mu(\widetilde{H} \cup \widetilde{U}), & \text{by Equation \ref{eqn:nonbipedcs1}}  \\
    & \geq (2/3-\epsilon) \mu(\widetilde{H} \cup \phi(M^*)) & \\
    & \geq (2/3-\epsilon) \mu_w(G_S), & \text{by Equation \ref{eqn:nonbipedcs2}.}
\end{align*}
\end{proof}
Finally, we complete the proof of Theorem \ref{thm:main_stream}.

\begin{proof}[Proof of Theorem \ref{thm:main_stream}]We apply the reduction in Theorem \ref{thm:reduction} to Algorithm \ref{alg:batch_bernstein}. By Proposition \ref{prop:batch_algo}, Algorithm \ref{alg:batch_bernstein} computes a $(2/3- \epsilon)$-approximate maximum cardinality matching using space $O(n \log n \poly(b/\epsilon))$ in the $b$-batch random-order stream model. 
It remains to show that if the input graph is of the form $\phi(G)$ for some weighted graph $G$ with batches $\mathcal{B} = \{ \phi(e): e \in G\}$, then $H\cup U$ is $(2/3-\epsilon)$-refolding approximate. Let $G_{late} \subseteq G$ denote the weighted edges corresponding to $\mathcal{B}_{late}$. An application of the Chernoff bound for negatively associated random variables (Theorem \ref{thm:chernoff_neg}) shows that 
$\Pr[\mu_w(G_{late}) \geq (1-2 \epsilon)\mu_w(G)] \geq 1-n^{-5}$ (the argument is similar to Lemma \ref{lemma:whp_stream}, only replacing the role of batches with weighted edges). Applying Lemma \ref{lemma:refolding_approx} to the graph $G$ and the subgraph $G_S := G_{late}$ yields
\begin{align*}
   \mu_w(\mathcal{R}(H\cup U)) & \geq (2/3 - \epsilon) \mu_w(G_{late}), &\text{by Lemma \ref{lemma:refolding_approx}} \\
   &\geq (1-2 \epsilon) (2/3 - \epsilon) \mu_w(G) & \\
   & \geq (2/3 - 3\epsilon) \mu_w(G),&
\end{align*}
as required.  Re-scaling $\epsilon$ and applying Theorem \ref{thm:reduction} yields the result. 
\end{proof}
\pagebreak
\section{$5/6$-Approximation in the Robust Communication Model}\label{sec:Comm}
In this section, we prove Theorem \ref{thm:main_2party} and Theorem \ref{thm:main_kparty}. By applying the results from the previous section, we can generalize the protocol of Azarmehr and Behnezhad \cite{RobustComm} to the weighted case. 
By the reduction in Theorem \ref{thm:space}, we can assume that the edge weights take integral values in $[W]$, for a large constant $W$. We will now describe the protocol for the two-party model. 

Let $\epsilon >0$ be the final parameter we are aiming for, and let 
\[ \lambda = \frac{\epsilon}{2048}, \beta = 144\lambda^{-4}\log(2W/\lambda).\]
Let $E_A$ denote the set of edges assigned to Alice and $E_B$ the set of edges assigned to Bob. Alice simulates a random-order stream. She unfolds the edges and runs Algorithm \ref{alg:batch_bernstein} on the corresponding unweighted $W$-batch random-order stream. That way, she obtains a set $H \subseteq \phi(E_A)$ with bounded edge degree $\beta$ and a set $U_A \subseteq \phi(E_A)$ consisting of all $(\phi(E_A  \setminus E_{early}), H, \beta, \lambda)$-underfull edges, where $E_{early} \subseteq E_A$ denotes the first $\frac{\epsilon}{W}m$ weighted edges in her simulated stream. She communicates $\mathcal{R}(H \cup U_A)$ to Bob. Bob outputs a maximum weight matching in $\mathcal{R}(H \cup U_A) \cup E_B$. See Algorithm \ref{alg:protocol} for a formal description. \\

\begin{algorithm}[H]
\caption{Robust Communication Protocol for Weighted Graphs}\label{alg:protocol}
\DontPrintSemicolon \LinesNotNumbered
    \SetAlgoNoEnd
    \begin{enumerate}
    \item  Alice simulates a random-order stream by ordering the edges in $E_A$ uniformly at random. 
    \item  Alice obtains $H \cup U_A \subseteq \phi(E_A)$ by running Algorithm \ref{alg:batch_bernstein} on input $\phi(E_A)$ with batches $\mathcal{B} = \{ \phi(e) : e \in E_A\}$. She communicates $\mathcal{R}(H \cup U_A)$ to Bob. 
    \item Bob outputs a maximum weight matching in $\mathcal{R}(H \cup U_A) \cup E_B$. 
    \end{enumerate}
 \end{algorithm}

The protocol for $k$ parties is similar, only that now all of the first $k-1$ parties should simulate a random-order stream  (we describe the protocol more formally in the proof of Theorem \ref{thm:main_kparty}). 

 Assume that each edge is assigned to Bob with probability $p \leq \frac{1}{2}$ (this will make the analysis applicable to the $k$-party setting). Let $U_B$ be the set of all $(\phi(E_B), H, \beta, \lambda)$-underfull edges, i.e. the set of underfull edges assigned to Bob. Let  $U:= U_A \cup U_B$ denote the set of all underfull edges. We will define an auxiliary fractional matching $x$ on $\mathcal{R}(H \cup U)$ of weight at least $(2/3-\epsilon) \mu_w(G)$.  We will then extend it to a fractional matching $y$ on $E_B \cup \mathcal{R}(H \cup U_A)$, and show that due to the additional edges in $E_B$, the fractional matching $y$ has weight at least $(5/6-\epsilon)\mu_w(G)$. 
 
Let $E_{late}:= E \setminus E_{early}$. Fix a maximum weight matching $M^*$ in $E_{late}$. Define a fractional matching $x$ on $\mathcal{R}(H \cup U)$ as follows: 

\begin{itemize}
	\item Start with $H_1 = H$ and $U_1 = U$. 
	\item For $i = 1, \dots, \lambda \beta: $
	\begin{itemize}
		\item Let $M_i$ be a maximum weight matching in $\mathcal{R}(H_i \cup U_i)$. 
		\item Let $H_{i+1} = H_i \setminus \phi(M_i\setminus M^*)$ and $U_{i+1} = U_i \setminus \phi(M_i \setminus  M^*).$
	\end{itemize}
	\item For every edge $e$, let $x_e = \frac{| \{i : e \in M_i\}|}{\lambda \beta}.$
\end{itemize}  In other words, we start with $H \cup U$, and then in each iteration, we find a maximum weight matching $M_i$ in the refolding and remove the edges in $\phi(M_i \setminus M^*)$ from $H \cup U$. 

\begin{remark}\label{remark:x}
Note that this is a valid fractional matching, since
$$    x_u = \sum_{e \ni u} x_e = \sum_{e \ni u} \frac{| \{i : e \in M_i\}|}{\lambda \beta} =  \sum_i \frac{|\{e \ni u: e \in M_i\}|}{\lambda\beta} \leq 1.  $$

Furthermore, note that $x_e \leq \frac{1}{\lambda \beta}$ whenever $e \notin M^*$. This is because, if $e \in M_i \setminus M^*$ for some $i$, then $e \notin \mathcal{R}(H_j \cup U_j)$ for all $j >i$. 
\end{remark}

\begin{lemma}\label{lemma:xval}
	It holds that $\sum_e w_e x_e \geq (\frac{2}{3} -  \epsilon) \mu_w(E_{late})$. 
\end{lemma}
\begin{proof}
For each $i$, let $G_i := E_{late} \setminus \left( \cup_{j < i} M_j \setminus M^* \right)$. We will apply Lemma \ref{lemma:refolding_approx} to the graph $G  \setminus \left( \cup_{j < i} M_j \setminus M^* \right)$ and subgraph $G_S =G_i.$
    Recall that we obtain $H_{i+1}$ from $H_{i}$ by removing the edges in $\phi(M_i \setminus M^*)$. Since $\phi(M_i \setminus M^*)$ is a matching, the degree of each edge in $\phi(G)$ will decrease by at most two in each iteration. Therefore,  $U_i$ contains all the edges in $G_i \setminus H_i$ that have $H_i$ degree less than $(1-\lambda)\beta - 2(i -1) \geq (1-3 \lambda ) \beta$. Therefore, by Lemma \ref{lemma:refolding_approx}, we have 
    \begin{equation}\label{eqn:w(Mi)}
    w(M_i) = \mu_w\left(\mathcal{R}(H_i \cup U_i)\right) \geq \left( \frac{2}{3}-\epsilon\right)\mu_w(G_i). 
    \end{equation}
    Also, $G_i$ is constructed so that it always contains $M^*$, so 
    \begin{equation}\label{eqn:mu(Gi)}
    \mu_w(G_i) \geq w(M^*) = \mu_w(E_{late}).
    \end{equation}
    Combining, we obtain 
    \begin{align*}
    \sum_{e \in \mathcal{R}(H \cup U)} w_e x_e &= \sum_{e \in \mathcal{R}(H \cup U)} w_e  \frac{| \{i : e \in M_i\}|}{\lambda \beta} &  \\
    & = \frac{1}{\lambda \beta} \sum_{i} \sum_{e \in \mathcal{R}(H \cup U)} w_e \mathbbm{1} \{e \in M_i\} &\\
    & = \frac{1}{\lambda \beta} \sum_i w(M_i) & \\
    &\geq \frac{1}{\lambda \beta} \sum_i \left( \frac{2}{3}-\epsilon\right)\mu_w(G_i). & \text{by Equation \ref{eqn:w(Mi)}}\\
    &\geq  \left(\frac{2}{3}-\epsilon \right) \mu_w(E_{late}),& \text{by Equation \ref{eqn:mu(Gi)}.}
    \end{align*}
\end{proof}

Recall that the set of edges that Bob has access to is $E_B \cup \mathcal{R}(H \cup U)$. We need to show that $\mu_w(E_B \cup \mathcal{R}(H \cup U) ) \geq( \frac{5}{6} - \epsilon) \mu_w(G)$. We will do this by extending the fractional matching $x$ on $\mathcal{R}(H \cup U)$ to a fractional matching $y$ on $E_B \cup  \mathcal{R}(H \cup U)$.  In order to describe $y$, we will condition on the set of early edges $E_{early}$, thereby fixing $\mathcal{R}( H \cup U)$ and $x$. For each edge $e \in E_{late}$, we have 
$$\Pr[e \in E_B | e \in E_{late}] =\frac{ \Pr[e \in E_B  \land e \in  E_{late}]}{\Pr[e \in E_{late}]} = \frac{p}{1-\epsilon/W}. $$
 and 
 $$\Pr[e \in E_A | e \in E_{late}] = 1 -  \frac{p}{1-\epsilon/W}. $$
Recall that $M^*$ is a fixed maximum weight matching in $E_{late}$.  Let $M_{in} := M^* \cap \mathcal{R}(H \cup U)$ and let $M_{out}:= M^* \setminus  \mathcal{R}(H \cup U). $
After drawing $E_B$, define a random matching $M' \subseteq M^*$ as follows: 
\begin{itemize}
    \item Include each edge $e \in M_{in}$ independently with probability $p$.
    \item Include each edge $e \in M_{out} \cap E_B$ independently with probability $1-\epsilon/W.$
\end{itemize}
Conditioned on $E_{early}$, each edge in $M_{out}$ ends up in $M'$ independently with probability $(1-\epsilon/W) \cdot  \frac{p}{1-\epsilon/W} = p$.  Each edge in $M_{in}$ also ends up in $M'$ independently with probability $p$, so overall each edge in $M^*$ ends up in $M'$ independently with probability $p$. 

For any edge $e \notin M^*$, let $p_e$ denote the probability that $e$ is \emph{not} adjacent to any edge in $M'$. In other words, 
\begin{equation*}
p_e =
\begin{cases} (1-p) & \text{if $e$ has exactly one endpoint matched by $M^*$,} \\
(1-p)^2 & \text{if both of the endpoints of $e$ are matched by $M^*$.} 
\end{cases}
\end{equation*}
We can now define   $\hat{y}$ on $E_B \cup \mathcal{R}(H \cup U)$.
\begin{equation*}
\hat{y}_e = 
\begin{cases}
	1  &\text{if $e \in M'$}, \\
	x_e &\text{if $e \in M^* \setminus M'$}, \\
	0  &\text{if $e \notin M^* $ and $e$ is adjacent to at least one edge of $M'$ }\\
	(1-p) \frac{x_e}{p_e} &\text{if $e \notin M^*$ and $e$ is not adjacent to $M'$}. 	 
\end{cases}
\end{equation*}
Finally, we scale down $\hat{y}$ and zero out some of the entries to obtain a valid fractional matching $y$. 
\begin{equation*}
	y_{(u,v)} = 
	\begin{cases} 
		0 & \text{ if $\hat{y}_u/(1+\epsilon) > 1$ or $\hat{y}_v/(1+\epsilon) > 1$ }  \\
		\frac{\hat{y}_{(u,v)}}{1+\epsilon} & \text{otherwise}. 
	\end{cases}
\end{equation*}	
\begin{lemma}\label{lemma:fractmatching}
	It holds that $$\mathbb{E}\left[\sum_{e \in E} w_e y_e \right] \geq \left(\frac{2}{3}+\frac{p}{3}-4 \epsilon \right) \mu_w(G).$$
	\end{lemma}
 The proof is similar to Lemma 4.6 in \cite{RobustComm}, and is included in Appendix \ref{appendix:RobustCommProofs}. Next, we round $y$ to an integral matching.
 
	\begin{lemma}
		There exists a matching of weight at least  $(1-3\epsilon)\sum_{e \in E} w_e y_e$ in $E_B \cup \mathcal{R}(H \cup U_A)$. 
	\end{lemma}
	\begin{proof}
		For every edge $e \notin M^*$, we have that $y_e \leq \frac{2}{\lambda \beta} \leq \epsilon^3$, by Remark \ref{remark:x}. So for any $S \subseteq V$ with $|S| \leq \frac{1}{\epsilon}$, it holds that 
		$$\sum_{e \in G[S]} y_e = \sum_{e \in G[S] \cap M^* } y_e +  \sum_{e \in G[S] \setminus M^* } y_e \leq |G[S] \cap M^* | + \frac{\epsilon^3}{\epsilon^2} \leq \left \lfloor \frac{|S|}{2} \right \rfloor + \epsilon. $$
		Applying Proposition \ref{prop:blossom} to $(1-2\epsilon)y$ yields the result. 
	\end{proof}
 Finally, we show that we have a large matching with high probability, and not just in expectation. 
	\begin{lemma}\label{lemma:whp}
		With probability at least $1-n^{-5}$, there exists a matching of weight at least  $\left(\frac{2}{3}+\frac{p}{3}-O(\epsilon) \right) \mu_w(G)$ in $E_B \cup \mathcal{R}(H \cup U_A)$.
	\end{lemma}
 The proof is similar to Lemma 5.2 in \cite{RobustComm}, and is included in Appendix \ref{appendix:whpProof}. We now complete the proofs of Theorem \ref{thm:main_2party} and Theorem \ref{thm:main_kparty}.
		
\begin{proof}[Proof of Theorem \ref{thm:main_2party}]
Suppose that the edge weights are in $[W]$. By Proposition \ref{prop:batch_algo}, Protocol \ref{alg:protocol} uses $O(n \log n \poly(W/\epsilon))$ words of communication with high probability. By Lemma \ref{lemma:whp}, the protocol achieves a $\left(\frac{2}{3}+\frac{p}{3}-O(\epsilon)\right)$-approximation with high probability. So by Theorem \ref{thm:space}, there exists a protocol that achieves a $\left(\frac{2}{3}+\frac{p}{3}-O(\epsilon)\right)(1-\epsilon)$-approximation using space $O(n \log n \log R)$ when the edge weights are in $\mathbb{R}^+$. Letting $p = \frac{1}{2}$ and re-scaling $\epsilon$ proves the theorem. 
\end{proof} 
\begin{proof}[Proof of Theorem \ref{thm:main_kparty}]
Suppose that the edge weights are in $[W]$. We need to adjust the protocol to the $k$-party model. The first party simulates the start of a random-order stream by selecting an ordering of their edges uniformly at random.  They unfold the edges and run Algorithm \ref{alg:batch_bernstein} on the corresponding unweighted $W$-batch random-order stream. They pass the memory state of the algorithm to the next party. Each of the next $k-2$ parties will continue to simulate the random-order stream and pass on the memory state of the algorithm. The $(k-1)$st party communicates $\mathcal{R}(H \cup U)$ to the last party, where $H \cup U$ is the unweighted graph computed by Algorithm \ref{alg:batch_bernstein} on the unfolded $W$-batch stream. Finally, the last party will output the maximum weight matching in the graph consisting of all edges to which they have access. That way, we can set $p=\frac{1}{k}$ and treat the first $k-1$ parties as Alice and the last party as Bob.  
 By Proposition \ref{prop:batch_algo}, the protocol uses $O(n \log n \poly(W/\epsilon))$ words of communication with high probability. By Lemma \ref{lemma:whp}, the protocol achieves a $\left(\frac{2}{3}+\frac{p}{3}-O(\epsilon)\right)$-approximation with high probability. So by Theorem \ref{thm:space}, there exists a protocol that achieves a $\left(\frac{2}{3}+\frac{p}{3}-O(\epsilon)\right)(1-\epsilon)$-approximation using space $O(n \log n \log R)$ when the weights are in $\mathbb{R}^+$. Re-scaling $\epsilon$ proves the theorem. 
\end{proof}

\section*{Acknowledgements}We thank Ola Svensson and Michael Kapralov for helpful discussions. We additionally thank Michael Kapralov for useful comments on the manuscript. 

\nocite{*}
\bibliographystyle{alpha}
\bibliography{bibliography.bib}
\appendix 
\section{Omitted Proofs from Section \ref{sec:Comm}}
\subsection{Proof of Lemma \ref{lemma:fractmatching}}\label{appendix:RobustCommProofs}
Throughout this section, we condition on $E_{early},$ thereby fixing $H$, $U$, $M^*$, $\hat{x}$ and $x.$
\begin{claim} \label{claim:hat_exp}
		For every vertex $u \in G$, it holds that 
		\begin{equation*}
	 \mathbb{E}\left[\sum_{e  \ni u} w_e \hat{y}_e\right] =
		\begin{cases}
		 (1-p) \sum_{e \ni u} w_e x_e  \quad & \text{if $u$ is not covered by $M^*$} \\
		  (1-p) \sum_{e \ni u} w_e x_e   +p  w_{e^*}  \quad & \text{if $u \in e^{*}$ for some  $e^* \in M^*$}. \\
		 \end{cases}
		\end{equation*}
	\end{claim}
	\begin{proof}
		First, consider the case when $u$ is not covered by $M^{*}$. For each edge $e$ adjacent to $u$, with probability $p_e$, the edge $e$ is not adjacent to $M'$, in which case $\hat{y}_e = \frac{(1-p)x_e}{p_e}$. Otherwise, with probability $1-p_e$, the edge $e$ is adjacent to $M'$, in which case $\hat{y}_e = 0.$ So 
	$$\mathbb{E}\left[\sum_{e  \ni u} w_e \hat{y}_e\right] =  \sum_{e \ni u} w_e p_e \frac{(1-p)x_e}{p_e} = (1-p) \sum_{e \ni u} w_e x_e .$$
	Now consider the case when $u$ is adjacent to some edge $e^* \in M^*$. If $e^* \in M'$, then $\hat{y}_{e^*} = 1$ and all edges adjacent to $e^*$ have value $0$. So 
	$$\mathbb{E}\left[\sum_{e  \ni u} w_e \hat{y}_e \Bigr| e^{*} \in M'\right] = w_{e^*}.$$
	If instead   $e^* \notin M'$, then $\hat{y}_{e^*} = x_{e^*}.$ For any other edge $e$ adjacent to $u$, the probability that $e$ is not adjacent to $M'$ is now $\frac{p_e}{1-p}$, in which case $\hat{y}_e = (1-p) \frac{x_e}{p_e}$. Otherwise, with probability $1-\frac{p_e}{1-p}$, the edge $e$ is adjacent to $M'$ and $y_e = 0$.
	So 
	$$\mathbb{E}\left[\sum_{e  \ni u} w_e \hat{y}_e \Bigr|  e^{*} \notin M'\right] = w_{e^*}x_{e^*} + \sum_{ e \ni u: e \neq e^*}  \frac{p_e}{1-p} w_e(1-p) \frac{x_e}{p_e} = \sum_{e \ni u} w_e x_e.$$
	Combining, we obtain 
	$$\mathbb{E}\left[\sum_{e  \ni u} w_e \hat{y}_e \right] = p \mathbb{E}\left[\sum_{e  \ni u} w_e \hat{y}_e \Bigr| e^{*} \in M'\right]  + (1-p)\mathbb{E}\left[\sum_{e  \ni u} w_e \hat{y}_e \Bigr| e^{*} \notin M'\right] = p w_{e^*} + (1-p) \sum_{e \ni u} w_e x_e.$$
	\end{proof}
	\begin{claim}\label{claim:hat_ratio}
		$$\mathbb{E}\left[\sum_{e \in E} w_e \hat{y}_e\right]  = \left(\frac{2}{3}+ \frac{p}{3}-\epsilon\right) \mu_w(E_{late}). $$
	\end{claim}
	\begin{proof}
		We have 
		\begin{align*}
			\mathbb{E}\Bigr[\sum_{e \in E} w_e \hat{y}_e\Bigr]  & =  \frac{1}{2} \sum_{u \in V} \mathbb{E}\Bigr[\sum_{e \ni u} w_e \hat{y}_e\Bigr] & \\
			& = \frac{1}{2} \Bigr[ \sum_{u \in V} (1-p) \sum_{e \ni u} w_e x_e + 2p \sum_{e^{*} \in M^{*} }w_{e^*}\Bigr] , & \text{by Claim \ref{claim:hat_exp} }\\
			&=   (1-p)\sum_{e \in E}w_e x_e + p \sum_{e^* \in M^*}w_{e^*} & \\
			& \geq   (1-p)\left(\frac{2}{3}-\epsilon\right) \mu_w(E_{late}) + p \mu_w(E_{late}), & \text{by Lemma \ref{lemma:xval}} \\
            &\geq  \left(\frac{2}{3}+ \frac{p}{3}-\epsilon\right) \mu_w(E_{late}). 
			\end{align*}
	\end{proof}
Next, we argue that we do not lose too much when transforming $\hat{y}$ into the valid fractional matching $y$. 
\begin{claim}\label{claim:hat_prob_bd}
	 For every vertex $u$, it holds that $\Pr[\hat{y}_u > 1+\epsilon] \leq \epsilon.$
\end{claim}

\begin{proof}
 Suppose first that $u$ is not covered by $M^*$. For each edge $e \ni u$, it holds that $p_e \geq (1-p)$, since $e$ has at most one neighbour in $M^*$. So we have
 $$ \hat{y}_u = \sum_{e \ni u} \hat{y}_e \leq \sum_{e \ni u} (1-p) \frac{x_e}{p_e} \leq \sum_{e \ni u} x_e \leq 1.$$
 Now suppose that $u$ is covered by $M^*$. Let $e^* \in M^*$ be the edge containing $u$. If $e^* \in M'$, then $\hat{y}_u =1$ and we are done. So we will condition on the event that $e^* \in M^* \setminus M'$. In this case, $\hat{y}_{e^*} =x_{e^*}$. 
Let $X:= \hat{y}_u - \hat{y}_{e^*} = \sum_{e \ni u: e \neq e^*} \hat{y}_e$. We will decompose $X$ into the sum of independent bounded random variables, and use the Chernoff bound to show that with high probability, $X$ is sufficiently small.  

Fix an edge $e = (u,v) \neq e^*$ adjacent to $u$. If $v$ is not matched by $M^*$ to another neighbour of $u$, then $y_e$ is independent of all the other edges adjacent to $u$. In this case $\hat{y}_e$ takes value $(1-p)\frac{x_e}{p_e}$ with probability $\frac{p_e}{1-p}$, and value $0$ otherwise. Note that $(1-p)\frac{x_e}{p_e} \leq \frac{2}{\lambda \beta} $ by Remark \ref{remark:x}. 
Otherwise, suppose that $(u,v)$ is matched by $M^*$ to a vertex $w$ in the neighbourhood of $u$. Let $e' = (u,w)$. Then $y_e + y_{e'}$ is independent of all the other edges adjacent to $u$. It takes value $(1-p)\frac{x_e+x_{e'}}{(1-p)^2} = \frac{x_e+x_{e'}}{1-p}$ with probability $\frac{p_e}{1-p} =\frac{p_{e'}}{1-p} = (1-p)$, and value $0$ otherwise. Note that  $\frac{x_e+x_{e'}}{1-p} \leq \frac{4}{\lambda \beta}$  by Remark \ref{remark:x}.  

By pairing together the edges that are matched together, we can now write $X$ as a sum of independent variables taking values in $[0, 4/\lambda \beta]$. We have $\mathbb{E}[\hat{y}_u -y_{e^*}] = \sum_{e \ni u : e \neq e^*} \frac{p_e}{1-p}\cdot (1-p) \frac{x_e}{p_e} \leq x_u \leq 1. $ So by Chernoff bound, we obtain 
\begin{align*}
\Pr[\hat{y}_u > 1 + \epsilon] & \leq \Pr[X > \mathbb{E}[X] + \epsilon] & \\
& = \Pr\left[ X \cdot \frac{\lambda \beta} {4} > \mathbb{E}\left[X \cdot \frac{\lambda \beta} {4}\right] + \epsilon \cdot \frac{\lambda \beta} {4}\right]&  \\
&\leq \exp \left[ - \frac{\epsilon^2 \lambda \beta}{12}\right]&  \\
& \leq \epsilon , \qquad & \text{since $\beta \geq 12 \epsilon^{-3}\lambda^{-1}$.}
\end{align*}

\end{proof}
\begin{claim} \label{claim:hat_to_y} 
	$$\mathbb{E} \left[ \sum_{e\in E} w_e y_e \right] \geq  (1-3\epsilon) \mathbb{E} \left[ \sum_{e \in E} w_e \hat{y}_e \right]  $$
\end{claim}
\begin{proof} 
We can think of the procedure of generating $y$ from $\hat{y}$ as two separate steps: First, in step 1, we zero out all edges $(u,v)$ with $\hat{y}_u > 1+ \epsilon$ or $\hat{y}_v > 1+ \epsilon$. Then, in step 2, we scale everything down by a factor of $(1+\epsilon)$.

Let us first consider step 1. For each vertex $u$, we zero out all the edges adjacent to $u$ with probability  $\Pr[\hat{y}_u > 1+ \epsilon]$. By Claim \ref{claim:hat_prob_bd}, we have $\Pr[\hat{y}_u > 1+ \epsilon] \leq \epsilon$. So in expectation, in step 1, we lose at most 
$$ \epsilon \sum_{u \in V} \sum_{v \in N(u)} w_{(u,v)}\hat{y}_{(u,v)} = 2 \epsilon \sum_{e\in E} w_e \hat{y}_e.$$ Then in step 2, we scale everything by a factor of $\frac{1}{1+\epsilon}. $
In total, we obtain
	$$\mathbb{E} \left[ \sum_{e\in E} w_e y_e \right] \geq \frac{1}{1+\epsilon} \sum_{e\in E} w_e \hat{y}_e  - 2\epsilon \sum_{e\in E} w_e \hat{y}_e  \geq  (1-3\epsilon) \sum_{e\in E} w_e \hat{y}_e. $$
\end{proof}
Finally, we lift the conditioning on $E_{early}$. 
\begin{proof}[Proof of Lemma \ref{lemma:fractmatching}]
		We have
		\begin{align*}
			\mathbb{E}\left[ \sum_{e \in E} w_e y_e \right] & = 	\mathop{\mathbb{E}}_{E_{early}}\left[ \mathop{\mathbb{E}}_{M'} \left[ \sum_{e \in E} w_e y_e \Bigr| E_{early} \right]  \right] & \\
			& \geq \mathop{\mathbb{E}}_{E_{early}} \left[  \mathop{\mathbb{E}}_{M'}  \left[ (1-3\epsilon) \sum_{e \in E} w_e \hat{y}_e \Bigr| E_{early} \right] \right], & \text{by Claim \ref{claim:hat_to_y}} \\
			&\geq \mathop{\mathbb{E}}_{E_{early}} \left[ (1-3\epsilon)\left(\frac{2}{3}+ \frac{p}{3}-\epsilon\right) \mu_w(E_{late})\right], & \text{by Claim \ref{claim:hat_ratio}} \\
			&= (1-3\epsilon)(1-\epsilon)\left(\frac{2}{3}+ \frac{p}{3}-\epsilon\right) \mu_w(G), & \text{since $\mathbb{E}[ \mu_w(E_{late}) ] = (1-\epsilon) \mu_w(G)$} \\
			& \geq \left(\frac{2}{3}+\frac{p}{3}-4 \epsilon \right) \mu_w(G). 
		\end{align*}
	\end{proof}
 
\subsection{Proof of Lemma \ref{lemma:whp}}\label{appendix:whpProof}
We will use the following concentration inequality for self-bounding functions. 
\begin{defn}[\cite{selfbounding}]
	A function $f: \{0,1\}^n \rightarrow \mathbb{R}$ is called self-bounding if  there exists functions $f_1, \dots f_n: \{0,1\}^{n-1} \rightarrow \mathbb{R}$ such that for all $x = (x_1, \dots x_n) \in \{0,1\}^n$,  
	$$0 \leq f(x) -f_i(x^{(i)}) \leq 1  \qquad \forall i \in [n],$$
		and
	$$\sum_{i=1}^n \left(f(x)-f_i(x^{(i)}) \right) \leq f(x),$$
	where $x^{(i)}$ is obtained by dropping the $i$-th entry of $x$. 
		\end{defn}
\begin{proposition}[\cite{selfbounding}]\label{prop:selfbound}
Let $X_1, \dots X_n$ be independent random variables taking values in $\{0,1\}$, and let $f: \{0,1\}^n \rightarrow \mathbb{R}$ be a self-bounding function. Define $Z = f(X_1, \dots X_n)$. Then 
	$$\Pr[Z \leq \mathbb{E}Z -t] \leq \exp\Bigr(\frac{-t^2}{2\mathbb{E}Z}\Bigr). $$
\end{proposition}

	\begin{proof}[Proof of Lemma \ref{lemma:whp}]
		Condition on $E_{early}$, thereby fixing $E_{late}, H$ and $U$. Let $e_1, \dots e_k$ denote the edges in $E_{late} \setminus \mathcal{R}(H \cup U)$. Bob has access to all of the edges $\mathcal{R}(H \cup U)$ with probability $1$, and to each of the edges $e_1, \cdots e_k$ independently with probability $ \frac{p}{1-\epsilon/W}.$ We will define a self-bounding function $f:\{0,1\}^k  \rightarrow \mathbb{N}$ describing the weight of the matching output by Bob. Given $x \in \{0,1\}^k $ , write $E_x := \{e_i : x_i = 1\}$ and define
		$$f(x):= \frac{\mu_w(\mathcal{R}(H\cup U)\cup E_x)}{W}. $$
		For $i = 1, \dots, k$, define 
		$$f_i(x^{(i)}) := f(x_1, \dots x_{i-1}, 0,  x_{i+1}, \dots, x_k ) =\frac{\mu_w(\mathcal{R}(H\cup U) \cup E_x \setminus \{e_i\})}{W}.$$ 
		Removing one edge from the graph reduces the weight of the matching by at most $W$, so 
		$$ 0 \leq f(x) - f_i(x^{(i)}) \leq \frac{W}{W} =1.$$
		Next, fix a maximum-weight matching $M_x$ of $\mathcal{R}(H \cup U) \cup E_x$. Then $f(x)$ and $f_i (x^{(i)})$ differ only if $e_i \in M_x$, in which case $f(x) - f_i (x^{(i)}) \leq \frac{w(e_i)}{W}$. Therefore, 
		$$\sum_{i=1}^k \left( f(x) - f_i (x^{(i)}) \right) \leq  \sum_{e \in M_x} \frac{w(e)}{W} =  \frac{\mu_w(\mathcal{R}(H\cup U) \cup E_x)}{W} = f(x).$$
		So $f$ is indeed self-bounding and we can apply Proposition \ref{prop:selfbound}. Let $X_i$ be the indicator variable that $e_i$ is assigned to Bob. Let $Z = f(X_1, \dots, X_k)$, and let $\mu = \mathbb{E}[Z] =   \frac{r \mu_w(G)}{W}$, where $r \geq \frac{2}{3}+\frac{p}{3}-4\epsilon$ is the approximation ratio that Protocol \ref{alg:protocol} achieves in expectation. By Proposition \ref{prop:selfbound}, we have
		$$\Pr \Bigr[Z \leq \frac{r\mu_w(G)}{W} - \sqrt{10r\mu_w(G)\log n/W} \Bigr] \leq \exp\Bigr( -\frac{10r\mu_w(G)\log n/W}{2 r\mu_w(G)/W } \Bigr)= \frac{1}{n^5}. $$
		Thus, with probability at least $1-\frac{1}{n^5}$, the protocol outputs a matching of size at least
		$r \mu_w(G) - W \sqrt{10r\mu_w(G)\log n/W} = (1-O(\epsilon))r  \mu_w(G)$ . The fact that the deviation is $O(\epsilon)$ follows from the assumption that $\mu_w(G) \geq O(W\epsilon^{-2} \log n)$. 
		\end{proof}
\end{document}